\documentclass[twocolumn]{autart}    

\let\theoremstyle\relax

\usepackage{amsmath}
\usepackage{amsthm}
\usepackage{amssymb}
\usepackage{graphicx}
\usepackage{enumerate}

\usepackage[linkcolor=blue,colorlinks=true,citecolor=blue]{hyperref}

\usepackage{arydshln}
\usepackage[round]{natbib}
\usepackage[titletoc]{appendix}

\newtheorem{theorem}{Theorem}[section]
\newtheorem{lemma}[theorem]{Lemma}
\newtheorem{proposition}[theorem]{Proposition}
\newtheorem{myalg}{Algorithm}

\begin{document}

\begin{frontmatter}

\title{Recursive Analytic Solution of Nonlinear Optimal Regulators
} 

\thanks[footnoteinfo]{
Corresponding author}

\author[a]{Nader Sadegh$^\star$}\ead{sadegh@gatech.edu},  
\author[b]{Hassan Almubarak}\ead{halmubarak@gatech.edu}    

\address[a]{Faculty of The George W. Woodruff School of Mechanical Engineering, Georgia Institute of Technology, Atlanta, GA, USA} 
\address[b]{School of Electrical and Computer Engineering, Georgia Institute of Technology, Atlanta, GA, USA}  


\begin{abstract}                          
The paper develops an optimal regulator for a general class of multi-input affine nonlinear systems minimizing a nonlinear cost functional with infinite horizon. The cost functional is general enough to enforce saturation limits on the control input if desired. An efficient algorithm utilizing tensor algebra is employed to compute the tensor coefficients of the Taylor series expansion of the value function (i.e., optimal cost-to-go). The tensor coefficients are found by solving a set of nonlinear matrix equations recursively generalizing the well-known linear quadratic solution. The resulting solution generates the optimal controller as a nonlinear function of the state vector up to a prescribed truncation order. Moreover, a complete convergence of the computed solution together with an estimation of its applicability domain are provided to further guide the user. The algorithm's computational complexity is shown to grow only polynomially with respect to the series order. Finally, several nonlinear examples including some with input saturation are presented to demonstrate the efficacy of the algorithm to generate high order Taylor series solution of the optimal controller.
\end{abstract}

\end{frontmatter}

\section{Introduction} \label{sec_Introduction}
Optimal control has been a leading methodology for both linear and nonlinear systems to address the optimization requirements quantitatively and qualitatively. Optimization problems, and consequently optimal controllers, can be constructed either as finite horizon or infinite horizon problems. Optimal control problems constructed as short horizon problems are usually tackled numerically for which different techniques have been proposed in the literature. A well known short horizon technique is known as Model Predictive Control (MPC) \citep*{mayne2014model,qin2003survey}. Using this method, the optimization problem needs to be solved at each point of time for the short horizon ahead and hence the solution is updated continuously. This method ensures the optimality of the solution. Despite its successful practicality, it can be very expensive computationally. On the other hand, infinite horizon optimal control problems can provide the optimal controller that needs to be updated with the current measurements, i.e. states, only. A leading approach that has untangled the infinite horizon optimal control problems involves solving the Hamilton-Jacobi-Bellman (HJB) equation. The HJB equation is a partial differential equation (PDE) and solving it for nonlinear systems is a daunting task. When the system is linear, however, and a quadratic cost functional is chosen, the HJB equation is reduced to the well-known Algebraic Riccati Equation (ARE) whose solution generates the prominent Linear Quadratic Regulator (LQR). Many different techniques have been proposed over the years to approximate the HJB equation's solution and/or approximate the associated optimal feedback control law for general nonlinear systems \citep*{adurthi2017sparse,al1961optimal,AlmubarakSadeghandTaylor2019,beard1998approximate,fujimoto2011stable,garrard1972suboptimal,garrard1977design,garrard1992nonlinear,kalise2018polynomial,lawton1998numerically,lukes1969optimal,sannomiya1971method,oishi2017numericalstablemani,sakamoto2008analytical,tran2017nonlinear,wernli1975suboptimal,xin2005new,yoshida1989quadratic}.

Recently, a new technique that utilizes the stable manifold theory was developed by \citet{sakamoto2008analytical} to approximate the stabilizing solution of the Lagrangian submanifold of the Hamiltonian system. The resulting algorithm has been successfully applied to control several nonlinear systems of academic interest such as an underactuated acrobot system \citep{horibe2016swing}, a constrained magnetic levitation system \citep*{tran2017nonlinear}, and an inverted pendulum with saturated input \citep{fujimoto2011stable}. The main drawback of this method is that it requires a large amount of a-priori information to produce an approximate suboptimal solution. Moreover, it is computationally complex and suffers from the curse of dimensionality.

Another approach is to successively approximate the solution of the HJB equation through iteratively solving a sequence of the linear generalized HJB (GHJB) equations. GHJB equations are linear PDEs that approach an approximate solution to the HJB equation starting by a randomly chosen feedback admissible control \citep{abu2005nearly,adurthi2017sparse,beard1998approximate,kalise2018polynomial,lawton1998numerically}. The successive approximation will eventually converge to an approximate solution of the HJB equation \citep{beard1998approximate}. A very popular method to successively approximate the solution is through using the Galerkin spectral method which was pioneered by \citet{beard1998approximate}. The main downside of this approach is that it does not lead to a single solution and the quality of the solution depends on the initialization of the control as well as the computation of many integrals.

Other methods try to solve the problem at each state in time. A well known algorithm generated from the HJB equation is known as the State-Dependent Riccati Equation (SDRE). The idea is to factorize the system's dynamics to put it in a form similar to the linear case but with a state dependent system's matrix. This factorization, also called apparent linearization \citep{cimen2008state}, leads to a state-dependent Riccati equation. Nevertheless, because apparent linearization is not unique and different linearizations result in different approximations, bad or nonconvergent solution is not unexpected. Moreover, no current method is known to produce an optimum factorization \citep{cimen2008state}.

The idea of using power series expansion has been investigated mostly in flight control systems papers. In an early work by \cite{al1961optimal}, Al’brekht studied nonlinear optimal control for analytic systems where he provided and proved a sufficient condition for optimiality. Additionally, he constructed a general systematic method to obtain the controller as a power series of the states for a scalar controller. \cite{lukes1969optimal} studied that for finite number of inputs and relaxed the analyticity assumption to twice differentiable. Besides, in a local sense, Lukes provided a proof of existence and uniqueness of the optimal control; nonetheless, neither presented a recursive closed form procedure to obtain the control law. In \cite{garrard1972suboptimal}, Garrard adopted a similar idea by expanding the value function as a power series of an artificial variable around the origin. The proposed method further developed in \cite{garrard1977design} and \citet{garrard1992nonlinear} by expanding the system's dynamics as a power series which is applicable to a wider class of systems. It was tested and compared to the LQR method in automatic flight control systems and proved its superiority, although it is only applicable to systems with low nonlinearites. \cite{sannomiya1971method} proposed a more efficient technique to find the coefficients of the series by assuming that the artificial variable is sufficiently small to find a sub-optimal control in a power series form of the artificial variable. \cite{yoshida1989quadratic} adopted Garrad's problem \citep{garrard1972suboptimal}, i.e. a nonlinear control affine system with a constant input matrix and a quadratic cost functional to develop a systematic method to construct a quadratic regulator for the finite and infinite horizon problems based on Pontryagin's minimum principle and they showed that the solution satisfies the HJB equation. In their development for the infinite horizon regulator, the gradient of the value function along the states was computed and used directly in the controller. This is only true, however, if the Jacobin of the value function gradient is symmetric, which one cannot assume but needs to enforce as we shall show in this paper. Moreover, no closed form of the unknown coefficients matrix was provided. Therefore, it is hard to develop a computer aided design, which they suggested as a future work. Recently, \cite{AlmubarakSadeghandTaylor2019} adopted the same problem and considered enforcing symmetry in computing the series coefficients which led to a closed form solution for the unknown coefficients.

In this paper, we solve the infinite horizon optimal control problem for nonlinear control affine systems. Our algorithm efficiently produces a Taylor series expansion of the optimal solution around an equilibrium point. As a consequence, we provide the optimal controller as a nonlinear function of all possible combinations of the states. Moreover, a recursive closed form formula is provided to obtain the control law as well as the value function. The proposed cost functional is general enough to incorporate not only the usual quadratic functions but also higher order state and input penalty terms. A byproduct of this generality is the ability to utilize input penalty function to enforce input saturation \citep{lyshevski1998optimal}. Finally, we analyze the convergence of the resulting power series and, more importantly, estimate its region of convergence. Subsequently, many limitations of the previous methods are overcome.

This paper is organized as follows. Section \ref{sec_Optimal Control Problem Statement} contains the problem statement as well as the basic theorems for the development of the nonlinear optimal controller through the HJB equation. The main development lies in Section \ref{sec_Infinite Horizon Nonlinear Regulator for Control Affine Systems} where we start by utilizing tensor algebra tools to construct a nonlinear matrix equation, which is an equivalent of the HJB equation, by efficiently representing the value function and the system nonlinearities as a multivariate Taylor series of the state variables. Thereafter, we present an algorithm to untangle this matrix equation recursively up to a prescribed truncation order as well as a general formula for computing the unknown coefficients matrix. In Section \ref{sec_Convergence Analysis and Radius Estimations}, analysis and a numerical estimation of the radius of convergence of the Taylor series are discussed. Afterwards, we present three simulation examples with different systems natures and nonlinearities in Section \ref{sec_Algorithm Implementation Examples and Simulation Studies}. Section \ref{sec_Conclusion} provides further insights, future directions, and other concluding remarks. Proofs of all Lemmas and Propositions are given in Appendices A-D.

\section{Optimal Control Problem Statement} \label{sec_Optimal Control Problem Statement}
Consider the nonlinear control affine dynamical system
\begin{equation} \label{gen.sys}
    \dot{x}= f(x)+ g(x) u
\end{equation}where $x \in \mathbb{R}^n$, $u \in \mathbb{R}^m$, $f: \mathbb{R}^n \rightarrow \mathbb{R}^n$ and  $g: \mathbb{R}^{n \times m} \rightarrow \mathbb{R}^n$. It is worth mentioning that the control affine requirement can be relaxed by introducing new states as shown by an example in \cite{AlmubarakSadeghandTaylor2019}. It is desired to find a control input, $u(t)$, that minimizes the cost functional
\begin{equation} \label{cost.func}
V(x(0),u(t))=\int_{0}^{\infty} Q(x(t))+R(u(t))  dt 
\end{equation}
where $Q: \mathbb{R}^n \rightarrow \mathbb{R}^+$ and $R: \mathbb{R}^m \rightarrow \mathbb{R}^+$ are the state and control penalty functions. For notational convenience, we drop the time argument, $t$, in our development. We shall use the notaion $C^k$ ($C^\omega$) to denote $k$-times continuously differentiable (analytic) functions on a neighborhood of the origin. Throughout the paper, we make the following assumptions on $f$, $g$, $Q$, and $R$:
\begin{description}
\item[A$_{1}$] $f(x)$ and $g(x)$ are at least $C^{2}$, $f(0)=0$ and $(F_1,G_0)$ is stabilizable where $F_1=\frac{\partial f}{\partial x}(0)$ and $G_0=g(0)$. Furthermore, $\|\frac{\partial f}{\partial x}\| \leq C \|x\|^\alpha$ and $\|\frac{\partial g_i}{\partial x}\| \leq C \|x\|^\alpha$, $i=1,\ldots,m$, for some positive $C$ and $\alpha$.
\item[A$_{2}$] $Q(x)$ is $C^\omega$, positive definite (i.e., $Q(x) >0$, $\forall x\neq 0$, and $Q(0)=0$) on $\mathbb{R}^n$, and $Q_1:=\frac{\partial^2 Q(0)}{\partial x^2} \geq 0$.
\item[A$_{3}$] $R(u)$ is $C^\omega$, even ($R(-u)=R(u)$), positive definite, and $R_1:=\frac{\partial^2 R(0)}{\partial u^2}>0$. Furthermore, $\rho(u) := (\frac{\partial R}{\partial u})^T$ has an inverse $\phi(v): = \rho^{-1}(v)$, and $\rho(0)=\phi(0)=0$.
\end{description}
 \textbf{Remark.} Assumption A$_1$ ensures stabilizabity of the linearized system and that $f$ and $g$ are well-posed . This assumption is essentialy the same as that in \cite{lukes1969optimal} for differentiable $f$ and $g$. Assumptions 2 and 3 are automatically satisfied for a quadratic (or any finite order polynomial) $Q(x)$ and $R(u)$. In particular, if $R(u)=\half u^TR_1u$, then $\rho(u)=R_1u$ and $\phi(v)=R_1^{-1}v$.
 
To state the necessary optimality condition, let the {\em Hamiltonian} associated with \eqref{gen.sys} and \eqref{cost.func} be denoted by
\begin{equation} \label{Hamiltonian}
H(x,V_x^*,u)= Q(x)+R(u)+ V^*_x (x) ^T \left( f(x)+g(x)u \right)
\end{equation}
where $V^*(x_0)=\min_{u\in L_2(0,\infty)} V(x_0,u)$ is the value function (i.e., optimal cost-to-go) and the vector $V^*_x=(\frac{\partial V^*}{\partial x})^T$. A necessary condition to minimize \eqref{cost.func} is that the well-known Hamilton-Jacobi-Bellman equation (HJB) be satisfied \citep[p. 278]{lewis2012optimal}: 
\begin{equation} \label{main_HJB}
HJB:= \min_u H(x,V_x^*,u) = 0 
\end{equation}
In the next subsection, we explore the conditions under which the HJB equation has a unique solution and show that its satisfaction is also sufficient for existence of an optimal controller.
\subsection{Existence and Uniqueness of the Optimal Solution and Sufficiency of the HJB Equation}
The objective of the optimal control problem specified by \eqref{gen.sys}-\eqref{cost.func} is to develop a nonlinear feedback control law that minimizes the cost functional \eqref{cost.func}. Such a feedback law will be shown to exists if Assumptions A$_1$-A$_3$ are satisfied. First, the following proposition shows that $\rho$ viewed as a vector field is {\em conservative} with its integral function denoted by $\psi$:
\begin{proposition} \label{psi and phi prop}
Suppose that $R$ satisfies A$_3$. There exists an analytic $(C^{\omega})$ function $\Psi \colon \mathbb{R}^m \rightarrow \mathbb{R}$ such that $(\frac{\partial \Psi}{\partial v})^T = \phi(v)= \rho^{-1}(v)$.
\end{proposition}
The following theorem, adopted from \citet{lukes1969optimal}, guarantees the existence and uniqueness of the value function $V^*$ and the optimal control. 
\begin{theorem} \label{theorem1}
Consider the optimal control problem \eqref{gen.sys}-\eqref{cost.func} satisfying A$_1$-A$_3$. There exists a unique continuously differentiable ($C^1$) optimal feedback control $u^* (x)$ in a neighborhood of the origin given by
\begin{equation} \label{opt.con1}
u^* (x)=- \phi(g^T V^*_x (x)) 
\end{equation}
where $V^*_x=(\frac{\partial V^*}{\partial x})^T$ satisfies the HJB equation:
\begin{equation} \label{HJB}
V_x^{*T} f(x)  - \Psi(g^T V_x^*(x)) +Q(x) =0
\end{equation} subject to $V^* (0)=0$. Furthermore, $V^*$, and consequently $u^*$, is $C^\omega$ if $f$ and $g$ are $C^\omega$.
\end{theorem}
\begin{proof} \label{existance theorem}
The proof of existence and uniqueness of the $C^1$ optimal controller and the associated $C^2$ value function $V^*(x)$ can be found in \citet{lukes1969optimal}. According to the HJB equation \eqref{main_HJB}, it is necessary for the optimal controller $u^*(x)$ to minimize the Hamiltonian $H(x,V_x^*,u)$ with respect to $u$. This implies $H_u(x,V_x^*,u)=\rho(u)+g^TV^*_x=0$. Since $\phi$ is the inverse of $\rho$, \eqref{opt.con1} follows. Letting $v^*=g^T V^*_x$, the corresponding HJB equation becomes
$$
H(x,V_x^*,u^*)= V^*_x (x)^Tf(x)+R(u^*)-v^{*T}\phi(v^*)+Q(x) 
$$
By Proposition \ref{psi and phi prop}, $v^{*T}\phi(v^*)-R(u^*)=\Psi(v^*)$ since $u^*=\phi(v^*)$. Thus \eqref{HJB} follows. The proof of the last part of the Theorem can be found in \citet{lukes1969optimal}.
\end{proof}
The preceding Theorem guarantees the existence and uniqueness of the optimal solution as well as the necessity of the HJB equation. The next Theorem establishes the sufficiency of the HJB equation \eqref{HJB} and is the infinite horizon version of the sufficiency results in \citet[p. 165-167]{liberzon2011calculus}.
\begin{theorem} \label{theorem2}
The optimal control problem \eqref{gen.sys}-\eqref{cost.func} has a unique minimizing solution $u^*$ given by \eqref{opt.con1} if the HJB equation \eqref{HJB} has a unique $C^2$ solution for $V^*$.
\end{theorem}
\begin{proof}
We first need to establish that $u^*$ given by \eqref{opt.con1} is a stabilizing controller or equivalently the origin of the closed loop system $\dot{x}=f(x)+g(x)u^*(x)$ is asymptotically stable. The HJB equation \eqref{main_HJB} with the optimal control $u^*(x)$ can be rewritten as  
$$
 V^*_x (x) ^T (f(x)+ g(x) u^*(x)) + Q(x)+R(u^*(x))=0
$$
Now, suppose that the value function, $V^*(x)$, satisfies the HJB equation.  Then,
\begin{gather*}
\begin{split}
\frac{dV^*}{dt} & = V^*_x (x) ^T (f(x)+ g(x) u^*(x)) \\ & = -  \big(Q(x)+R(u^*(x)) \big) \leq - Q(x) < 0, \; \forall x\neq 0
\end{split}
\end{gather*}
By Lyapunov stability theory \citep[Theorem 4.1]{khalil2002nonlinear}, the origin of the closed loop system, $f(x)+g(x)u^*(x)$, is asymptotically stable.
The proof of $u^*$ in \eqref{opt.con1} minimizing \eqref{cost.func} is completely parallel to that in \citet[p. 165-167]{liberzon2011calculus} and is omitted. 
\end{proof}
Therefore, satisfaction of the HJB equation is both necessary and sufficient for existence of the optimal control solution to \eqref{gen.sys} and \eqref{cost.func}. Moreover, in view of Theorem \ref{theorem1}, a unique optimal solution $u^*$ is attainable if Assumptions A$_1$-A$_3$ are satisfied. It is also important to note that the value function $V^*$ is real analytic (i.e., has a convergent Taylor series) if $f$ and $g$ are real analytic. 

In the next section, we will formulate an efficient method for computing the Taylor series of $V^*$ to within a prescribed order. Additionally, a closed form solution to find the coefficients of an arbitrary order will be provided. 



\section{Infinite Horizon Nonlinear Regulator for Control Affine Systems} \label{sec_Infinite Horizon Nonlinear Regulator for Control Affine Systems}
The key idea is to use tensor algebra tools  to compute the Taylor series of the value function as
\begin{equation} \label{opt.val.tensors}
V^*(x)=\sum_{k \geq 1} \frac{x^T \bar{P_k} x^{\otimes k}}{k+1}
\end{equation}
where $\bar{P_k} \in {\mathbb{R}}^{n\times n^k}$ is a matricized symmetric tensor of rank \(k\), \(\otimes\) is the Kronecker product and \(x^{\otimes k}=x\otimes x\ldots  \otimes x\) \(k\) times. A tensor \(\bar{P_k}\) is represented by a matrix by unfolding it. The following subsection highlights the key definitions and properties of multivariate tensors needed for further refinement of the value function $V^*(x)$ and the resulting optimal control $u^*(x)$.   
\subsection{Series Formulation of the Optimal Solution via Multivariate Tensors}
We begin by reviewing the key properties of Kronecker product \citep[Chapter 4]{horn1994topics}, which is frequently used throughout the remainder of the paper particularly in our main numerical algorithm. For any matrices $A$,  $B$, and $C$, and integers $m$ and $n$,
\begin{enumerate} [1)]
\item $ (A \otimes B)^T = A^T \otimes B^T $. 
\item $ (A \otimes B)^{-1} = A^{-1} \otimes B^{-1} $, for non-singular $A$ and $B$
\item $ \text{vec} (ABC) = (C^T \otimes A) \text{vec}(B)$, where vec is an operator such that vectorizing $A_{m\times n}=\begin{bmatrix} A_1 && A_2 && \dots \end{bmatrix}$ yields $\text{vec} (A_{m \times n}) = \begin{bmatrix} A_1^T && A_2^T && \dots  \end{bmatrix}_{mn \times 1}^T$, where $A_1,A_2,\dots$ are vectors of $A$. \end{enumerate}
It is crucial to mention that repeated terms created from the Kronecker product of the states vector generates dependent vectors in the tensor $\bar{P_k}$. To eliminate redundancy, the vector $x^{\otimes k} \in \mathbb{R}^{n^k}$ can be reduced to lexographic listing $x^k$ of $m_k:=\binom{n+k-1}{k}$ linearly independent terms of $x^{\otimes k}$ \citep{loparo1978estimating}: 
\begin{equation} \label{xk}
x^k=\begin{bmatrix}x_1^k \\ \sqrt{k}x_1^{k-1}x_2 \\ \vdots \\ c_{k_1,k_2,\dots,k_n} x_1^{k_1} x_2^{k_2} \dots x_n^{k_n} \\ \vdots \\ \sqrt{k} x_n^{k-1} x_{n-1} \\ x_n^k \end{bmatrix}    
\end{equation}
where $c_{k_1,k_2,\dots,k_n}^2=\frac{k!}{k_1!k_2!\ldots k_n!}$ are the multinomial coefficients representing the number of times $x_1^{k_1} x_2^{k_2} \dots x_n^{k_n}$ is repeated in $x^{\otimes k}$. Therefore, for each $k$, there exists a unique $L_k \in \mathbb{R}^{n^k \times m_k}$ such that $x^{\otimes k}=L_k x^k$ and $\bar{P}_kx^{\otimes k}=P_kx^k$ with $P_k= \bar{P_k} L_k$ to be the reduced matricized tensor. It follows that $L_k$ is orthonormal (i.e. $L_k^T L_k=I$) with entries that are either zero or $1/c_{k_1,k_2,\dots,k_n}$. For example, for $n = 2$,
$$
x\otimes x=x^{\otimes 2}= \begin{bmatrix} x_1^2 \\ x_1 x_2 \\ x_2 x_1 \\ x_2^2 \end{bmatrix}, \;
x^2= \begin{bmatrix} x_1^2\\ \sqrt{2}x_1 x_2\\x_2^2 \end{bmatrix}, \;
L_2=\begin{bmatrix} 1&0&0\\0&\frac{1}{\sqrt{2}}&0\\0&\frac{1}{\sqrt{2}}&0\\0&0&1 \end{bmatrix}
$$
The following Proposition shows that the higher dimensional tensors inherit scalar product and norm from their vector parents.
 \begin{proposition}\label{vectors power prop}
 For $x,y \in \mathbb{R}^n$, an integer $k \in \mathbb{N}$, and the lexographic listing vectors $x^k$ and $y^k$ defined in \eqref{xk}, $\langle x^k, y^k \rangle = (x^T y)^k$. In particular, $||x^k||=||x||^k$.
\end{proposition} 
Now, using the mapping $P_k=\bar{P}_k L_k$ in \eqref{opt.val.tensors} provides the value function in terms of the reduced matrices $P_k$'s:
\begin{equation} \label{opt.val}
V^*(x)=\sum_{k\geq1} \frac{x^T P_k x^k}{k+1}
\end{equation}
The gradient of the $k$-th term of $V$, i.e., $V_k(x)=\frac{1}{k+1}x^T P_k x^k$, with a symmetric Jacobian (e.g., $P_1=P_1^T$ for $k$=1) is given by
\begin{equation} \label{opt.grad_k}
\left(\frac{\partial V_k}{\partial x}\right)^T= P_k x^k
\end{equation}
The following proposition provides the necessary and sufficient symmetry conditions on $P_k$ in order for \eqref{opt.grad_k} to be valid.
\begin{proposition}\label{prop:sym}
The given relationship in \eqref{opt.grad_k} holds for an arbitrary tensor $P_{k}\in \mathbb{R}^{n\times m_{k}}$ if and only if one of the following equivalent conditions holds
\begin{enumerate}[i)]
\item The Jacobian matrix of $P_k x^k$ (i.e., the Hessian of $V_k$) is symmetric.
\item ${\textnormal{vec}} (P_k) = K_{k}^T p_k$ for some $p_k \in \mathbb{R}^{m_{k+1} }$ where $K_k \in \mathbb{R}^{n m_{k}\times m_{k+1} }$ is the unique matrix that reduces $x^k\otimes x$ to $x^{k+1}$: $x^k\otimes x=K_k^T x^{k+1}$.
\end{enumerate}
\end{proposition}
This is a very important proposition which is missed in \cite{yoshida1989quadratic}. Additionally, the second part of Proposition \ref{prop:sym} leads to a more computationally efficient methodology than the one in \cite{AlmubarakSadeghandTaylor2019}. To illustrate the symmetry condition, let us give an example of $P_2$ in the 2-dimensional ($n=2$) case. The Kronecker product $x^2\otimes x$ is reduced to $x^3=\begin{bmatrix}
x_1^3,\sqrt{3} x_1^2 x_2, \sqrt{3} x_1 x_2^2,x_2^3 \end{bmatrix}^T$ via $x^2\otimes x=K_2^T x^{3}$ with
$$
K_2=\begin{bmatrix}
1&0&0&0&0&0\\0&\frac{1}{\sqrt{3}}&\sqrt{\frac{2}{3}}&0&0&0\\0&0&0&\sqrt{\frac{2}{3}}&\frac{1}{\sqrt{3}}&0\\0&0&0&0&0&1
\end{bmatrix}
$$
The reduced $P_2$ satisfying the symmetry condition (i.e., vec$(P_2)=K_2^Tp_2$) is given by \(P_2= \begin{bmatrix} a&\sqrt{2}b&c\\b&\sqrt{2} c&d \end{bmatrix}\) for $p_2=[a, \sqrt{3} b, \sqrt{3} c, d]^T$. The cubic component of the value function corresponding to $P_2$ is
$$
V^*_3(x)=\frac{1}{3} x^TP_2x^2=\frac{a}{3}x_1^3+bx_1^2x_2+c x_1 x_2^2+\frac{d}{3}x_2^3
$$
The gradient of $V_3^*$ is easily seen to be $V^*_{x3}(x)=P_2x^2$ as expected. Note that for an an arbitrary (non-symmetric) $S_2\in \mathbb{R}^{2 \times 3}$, $S_2x^2$ is not a gradient of any cubic function. 
Once the symmetry condition for each $P_k$ holds, the value function's gradient can be expressed as 
\begin{equation} \label{grad.opt.val}
V_x^*(x)=  \sum_{k \geq 1} {P_k x^k}
\end{equation}
leading to the optimal control law 
\begin{equation} \label{exxpanded.optimal.control}
  u^*(x)=-\phi \Big(\sum_{k\geq 1} g(x)^TP_k x^k\Big)  
\end{equation}
The input penalty function $R(u)$ can also be used to enforce input constraint \citep{lyshevski1998optimal}. For instance, if $\phi(v)=[\tanh(v_1) \; \ldots \; \tanh(v_m)]^T$, which we will use in an example in Section \ref{sec_Algorithm Implementation Examples and Simulation Studies} to confront saturation, then each input is restricted to $[-1,1]$. The resulting $R$ and $\psi$ are
$\psi(v)=\sum_{i=1}^m \ln(\cosh(v_i))$ and $R(u)=\sum_{k \geq 1} \frac{\|u\|^{2k}_{2k}}{2(2k-1)k}$ where $\|u\|_p$ denotes the $p$-norm: $\|u\|_p^p=\sum_{i=1}^m |u_i|^p$.
\subsection{Nonlinear Regulator (NLR) Detailed Algorithms}

The first algorithm we present in this subsection, solves the HJB equation \eqref{HJB} sequentially by exploiting the Taylor series expansions of $V_x^*$ in \eqref{grad.opt.val}, the system vector field $f(x)=\sum_{j\geq 1}{F_j x^j}$, and other relevant functions:
\begin{equation} \label{HJBk}
\sum_{k,j \geq 1} x^{jT} F_j^T P_k x^k  - \Psi\Big(g^T \sum_{k \geq 1} P_k x^k\Big) +Q(x)=0
\end{equation} 
To solve for $P_k$, the coefficients of all the independent terms of order $k+1$ in \eqref{HJBk} are set to zero. By doing so, a set of linear matrix equations results that is solved for $P_k$ to render the value function up to a desired order of truncation. The procedure starts by solving the ARE for $P_1$. The subsequent $P_k$'s are evaluated recursively based on previous $P_k$'s and the other known data.
\\
\begin{myalg} \label{algorithm.main}
Recursive Closed Form Solution of the NLR
\end{myalg} 
Given analytic functions $f(x)$, $g(x)$, $Q(x)$, $R(u)$ satisfying A$_1$-A$_3$, and the approximation order $\bar{k}$, execute the following steps to compute $P_k$, $k=1,\ldots,\bar{k}$, used to find the Taylor expansion of the optimal cost-to-go $V^*(x)$ and the optimal control $u^*(x)$: 

1. Determine the matrix components
\begin{enumerate}[i.]
\item $F_k$ of $f(x)=\sum_{k=1}^{\bar{k}} {F_k x^k}+O(x^{\bar{k}+1})$,
\item  $G_0$ and $G_{ik}$ of $g_i(x)=g_{i0} + \sum_{k=1}^{\bar{k}} {G_{ik} x^k}+O(x^{\bar{k}+1})$ for $i=1,\ldots,m$, where $g_i$ is the $i$-th column of $g$,
\item  $Q_{k}$ of $Q(x)=\half x^T Q_1 x+ \sum_{k=2}^{\bar{k}} x^T Q_k x^k + O(x^{\bar{k}+2}) $,
\item $\tilde{R}_{k}$ of $\Psi(v)=\half v^T \tilde{R}_1 v + \sum_{k=2}^{\bar{k}}v^T \tilde{R}_k v^k + O(v^{\bar{k}+2})$,
where both $Q_k$ and $\tilde{R}_k$ are assumed to be symmetric tensors of order $k$. Note that $\tilde{R}_1=R_1^{-1}$. 
\end{enumerate}
The main operation needed to implement the algorithm is the dot product of two power series. To this end, let $[s]_k$ denote the $k$-th tensor coefficient of the power series $s(x)=\sum_{k\geq 0}S_k x^k$,i.e., $[s]_k=S_k$. The dot product $s^Tl(x)=l^Ts(x)$ of $s(x)$ with another power series $l(x)=\sum_{k\geq 0}L_k x^k$ is given by
$$
[s^Tl]_k=\sum_{j=0}^{k}K_{k-j,j} S_j^TL_{k-j}
$$
where the matrix $K_{i,j}$ reduces $x^i \otimes x^j$ to $x^{i+j}$: $(x^i \otimes x^j)^T=x^{(i+j)T}K_{i,j}$. Similarly, the product $g^Ts(x)$ of a matrix valued power series $g(x)=[g_1(x) \cdots g_m(x)]$ with $s(x)$ is evaluated from $s^Tg_i$, $i=1,\dots,m$. For convenience, we denote $K_{k,1}$ by $K_k$ for $k=1,2, \ldots,\bar{k}$.

2. For $k=1$, the quadratic component of the HJB equation \eqref{HJBk} yields the ARE
$$
F_1^T P_1+P_1 F_1 + Q_1 -P_1 G_0 R_1^{-1} G_0^T P_1=0
$$
Assumption A$_1$, guarantees that the solution $P_1$ to the ARE is a symmetric positive definite matrix.

3. For $k\geq2$, compute $P_k$ by collecting the $(k+1)$-th order components of the HJB equation \eqref{HJBk}:
\begin{align*}
 & x^T F_1^TP_k  x^k + x^TQ_kx^k - x^TP_1 G_0 R_1^{-1} G_0^T P_k x^k \\
 &  + ([f^Th_k]_{k+1}-[\Psi(g^Th_k)]_{k+1})x^{k+1}=0
\end{align*}
where $h_{k} (x) = \sum_{j=1}^{k-1}P_j x^j$. Interestingly, we have a negative feedback from the linear solution multiplied by the unknown matrix in the third term. Defining the closed-loop matrix $F_c=F_1- G_0R_1^{-1} G_0^TP_1$, we have
\begin{equation*}
x^T(F_c^TP_k+Q_k)x^k+[f^Th_k-\Psi(g^Th_k)]_{k+1}x^{k+1}=0
\end{equation*}
Vectorizing the preceding equation using the Kronecker product identities, we get  
\begin{equation*}
(x^k\otimes x)^T  \text{vec}(F_c^T P_k+Q_k) + [f^Th_k-\Psi(g^Th_k)]_{k+1}x^{k+1} =0
\end{equation*}
Again, the Kronecker product will create some repeated basis. As shown previously, there exists a reducer matrix $K_k \in \mathbb{R}^{nm_k\times m_{k+1}}$ such that $x^k \otimes x= K_k^T x^k$. Thus
\begin{equation*}
x^{(k+1)T} \Big( K_k (I \otimes F_c^T) \text{vec}(P_k) +[f^Th_k-\Psi(g^Th_k)]_{k+1}^T +q_k\Big) =0
\end{equation*}
where ${q}_{k} =K_k\text{vec} (Q_{k})$. Before solving for $P_k$, we need to enforce symmetry as given by Proposition \ref{prop:sym}: That is, $\text{vec}(P_k) = K_k^T p_k$ where $p_k \in \mathbb{R}^{m_{k+1}}$ is an unknown vector. Since $x^{k+1}$ is an arbitrary vector, it can be dropped from the preceding equation reducing it to 
$$M_k p_k+[f^Th_k-\Psi(g^Th_k)+q_k]_{k+1}^T+{q}_{k}=0
$$ where $M_{k}= K_k (I \otimes F_c^T) K_k^T$. Solving for $p_k$ yields:
\begin{align} \label{Vec_Pk}
\begin{split}
\text{vec}(P_k) = K_k^T M_{k}^{-1}  \left( [\Psi(g^Th_k)-f^Th_k]_{k+1}^T-q_{k} \right)
\end{split}
\end{align}
Finally, reshape $\text{vec}(P_k)$ into an $n \times m_k$ matrix to get $P_k$.

4. The optimal controller resulting from the computed value function at the end step $\bar{k}$ is given by
$$
u_{\bar{k}}^* (x)=-\phi \big( g(x)^T \sum_{k=1}^{\bar{k}}P_k x^k \big)
$$
Note that in the case of a quadratic cost, the controller will be reduced to $
u_{\bar{k}}^* (x)=- R_1^{-1} g(x)^T \sum_{k=1}^{\bar{k}}P_k x^k
$.

End of Algorithm.

The algorithm for most part is self explanatory. The main remaining issue that needs to be addressed is the validity of inverting $M_k$ required at each step $k\geq 2$, which is established by the next Lemma. 
\begin{lemma} \label{invM_lemma}
Matrix $M_{k}= K_k (I \otimes F_c^T) K_k^T$ is invertible, $\forall k \geq 2$. Furthermore, if $F_c$ is negative definite, then $\inf_{\|v\|=1} \|M_k v\| \geq \alpha^{-1}$ or equivalently $\|M_k^{-1}\| \leq \alpha$, $\forall k \geq 1$, with $\alpha=2/\lambda_{\min}(-F_c-F_c^T)$.
\end{lemma}
In spite of being invertible, there is no guarantee that the inverse of $M_k$ remains bounded as $k$ increases. In fact, a simple numerical example can be constructed to illustrate that $M_k^{-1}$ can grow unboundedly whenever $F_c$ is not negative definite. For instance,
$F_1=\begin{bmatrix}
0 & 1\\1 & 1
\end{bmatrix},
G=\begin{bmatrix}
0 \\ 1
\end{bmatrix}$, $Q_1=100I$, and $R_1=1$ produce gain $\kappa=R_1^{-1}G_0^TP_1=[11.050 \;\;  12.0950]$ and $F_c=F_1-G_0K_1=\begin{bmatrix} 0 & 1\\-10.050 & -11.095 \end{bmatrix}$ is sign indefinite. While $\|M_k\| \leq \|F_c\|$, $\forall k \geq 1$, the 2-norm of $M_k^{-1}$ for $k=5, 50, 100$ grows to $5.758$, $2.14 \times 10^4$, and $9.45 \times 10^7$, respectively. This indicates that the original system is poorly conditioned. Fortunately, this issue can be easily resolved by means of a {\em linear} state transformation. Intuitively, a linear transformation improves conditionality by scaling and rotating the state-space coordinate axes. The following Lemma, which complements Lemma \ref{invM_lemma}, guarantees the existence of such a linear transformation.
\begin{lemma} \label{invM_lemma2}
There exists a symmetric positive definite transformation matrix $T \in \mathbb{R}^{n \times n}$ with respect to which $\hat{F}_c=TF_cT^{-1}$ is negative definite and the resulting $\hat{M}_{k}= K_k (I_{m_k} \otimes \hat{F}_c^T) K_k^T$ satisfies the conclusions of Lemma \ref{invM_lemma}.  
\end{lemma}
As can be seen from the proof of Lemma \ref{invM_lemma2}, the transformation matrix $T=\sqrt{P_c}$  where $P_c$ satisfies the Lyapunov equation $F_c^T P_c+P_cF_c+I=0$. Following this procedure, $T=\begin{bmatrix} 1.025  &  0.223 \\ 0.223  &  0.223 \end{bmatrix}$ and $\alpha=4.608$ for our numerical example. The resulting $\|\hat{M}_k^{-1}\|$ for $k=5, 50, 100$, which are $1.623$, $3$, and $3.489$, respectively, stay well below $\alpha$ as expected. This finding necessitates a linear state transformation of the original system and cost function whenever $F_c$ is poorly conditioned. \\
\begin{myalg} \label{algorithm.modified}
Modified (NLR) Algorithm with guaranteed numerical conditionality
\end{myalg} 
 1. Execute steps 1 and 2 of Algorithm \ref{algorithm.main} and compute $F_c=F_1-G_0R_1^{-1}G_0^TP_1$ and $\bar{F}_c=-\half(F_c+F_c^T)$. If $\bar{F}_c>0$ and $\alpha=1/\lambda_{\min}(\bar{F}_c)$  is less than a certain threshold, proceed with the remaining steps of Algorithm \ref{algorithm.main}.  Else, continue.
 
 2. Solve the Laypunon equation $F_c^T P_c+P_cF_c+I=0$ for $P_c>0$ and set the linear transformation matrix to $T=\sqrt{P_c}$. Replace $f(x)$, $g(x)$, and $Q(x)$ by their transformed equivalents
$  f(x) \leftarrow Tf(T^{-1}x)$, $g(x)  \leftarrow Tg(T^{-1}x)$ and $Q(x) \leftarrow Q(T^{-1} x)$. Then, return to step 1 of Algorithm \ref{algorithm.main} with the new data.

End of Algorithm.

The next Theorem summarizes the main contributions of the preceding algorithms and their computational efficiency.
\begin{theorem}
The NLR Algorithms \ref{algorithm.main} and \ref{algorithm.modified} exactly compute each matrix component
$P_{k}\in%
\mathbb{R}
^{n\times m_{k}}$ of the value function $V^*(x)$ in \eqref{opt.val} in order to satisfy the HJB equation \eqref{HJBk} up to a prescribed order
$\bar{k}$. Each computed $P_{k}$ satisfies the symmetry condition required to produce the optimal control function \eqref{exxpanded.optimal.control}. Moreover, the computational complexity of each algorithm grows at most polynomially in $\bar{k}$ and Algorithm \ref{algorithm.modified} guarantees that the numerical procedure (i.e., inversion of $M_k$) used to compute $P_k$ is well-conditioned (i.e. the condition number $\|M_k\|\|M_k^{-1}\|$ of $M_k$ is uniformly bounded) regardless of $k$.
\label{thm:sum}
\end{theorem}
\begin{proof}
The first part of the theorem follows from the described steps in Algorithms \ref{algorithm.main} and \ref{algorithm.modified} together with Lemmas \ref{invM_lemma} and \ref{invM_lemma2}. To prove the growth order, it should be clear that the computational complexity of each term in \eqref{Vec_Pk} including $[\Psi(g^Th_k)]_{k+1}$ is at most polynomial since the length of $p_k$ is $m_{k+1}=O(k^n)$ and the complexity of $[(g^Th_k)^j]_{k+1}$, $j\leq k$, needed to compute the power series for $\Psi(.)$ is also of polynomial growth rate. Therefore, the overall complexity of the algorithm is at most polynomial in the truncation order $\bar{k}$, which completes the proof.
\end{proof}
\section{Convergence Analysis and Radius Estimations} \label{sec_Convergence Analysis and Radius Estimations}
The proposed method in this paper efficiently generates the exact Taylor expansion of the value function and the associated optimal control law for analytic systems. The main goal of this section will be to estimate the {\em region of convergence} (ROC) associated with the generated value function $V^*$ and the resulting optimal controller. Formally, we define the region of convergence ROC$(f)$ of an arbitrary analytic function $f$ to be the largest open neighborhood of the origin in which $f$ is analytic. The radius $r_{\rm con}(f)$ of convergence on the other hand is the radius of the largest $n$-sphere that can be inscribed in ROC$(f)$. The importance of estimating the ROC is that it allows the user of the algorithm to know where it works best. It should be pointed out that it is guaranteed that higher orders of approximation lead to a better performance inside the ROC. On the contrary, if we operate outside the ROC, the controller is only suboptimal without any guarantees of stability. For a single variable analytic function, it is well known (see for example \citet[Theorem IV.7]{flajolet2009analytic}) that $r_{\rm con}(f)^{-1}=\lim_{k\rightarrow \infty}\sup|[f]_k|^{\frac{1}{k}}$ where $[f]_k$ is $k$-th Taylor series coefficient of $f$. The next Theorem builds on the single variable case to formulate the ROC of the value function ($\text{ROC}(V^*)$) based on the limiting behaviour of the matrix coefficients $P_k$'s. 
\begin{theorem} \label{theorem.conv}
Let $P_k$ be the matrix coefficients generated by Algorithm \ref{algorithm.modified} and define the {\em directional} radius of convergence along a unit vector $\upsilon \in \mathbb{R}^n$ by
$$r^*_\upsilon =\Big(\lim_{k \rightarrow \infty}\sup \|P_k\upsilon^k\|^{\frac{1}{k}}\Big)^{-1}$$
Then, the value function $V^*(x)=\sum_{k\geq 1} \frac{1}{k+1} x^T P_k x^k$ on
$${\rm ROC}(V^*) = \{x \in \mathbb{R}^n: x=r\upsilon, \; \|\upsilon\|=1, \; 0 \leq r < r^*_\upsilon\}.$$
\end{theorem}
\begin{proof}
By Theorem \ref{theorem1}, the value function $V^*\in C^{\omega}$ subject to the analyticity requirement of Algorithm \ref{algorithm.modified} (i.e., $f,g,Q$, and $R \in C^{\omega}$). Let $x$ be a nonzero vector in ROC$(V^*)$. The restriction $V_x^*|_\upsilon(z)=V_x^*(z\upsilon)$ of $V_x^*$ to $x=z\upsilon$ for $\upsilon \in \mathbb{R}^n$, $\|\upsilon\|=1$, is also an analytic function of a single complex variable $z$ and has a convergent Taylor series on $\mathbb{D}_r=\{z\in\mathbb{C}: |z|<r\}$ for some $r>0$. Applying the radius of convergence formula to each component of $V_x^*|_\upsilon(z)$ yields $r\leq r^*_{\upsilon,\infty}$, where $r^*_{\upsilon,\infty}=(\lim_{k \rightarrow \infty} \sup \|P_k\upsilon\|_{\infty}^{1/k})^{-1}$. From $\|P\upsilon^k\|\leq \sqrt{n} \|P\upsilon^k\|_{\infty}$, it follows that $r^*_{\upsilon,\infty}\leq r_\upsilon^*$ so that $r \leq r_\upsilon^*$. Thus, ${\rm ROC} (V^*) \subset\Re^*$ where $\Re^*$ is the right side of ROC$(V^*)$ in the statement of the Theorem. 

To prove $\Re^* \subset$ ROC$(V^*)$, it is sufficient to show that $V^*(x)=\sum_{k\geq 1} \frac{1}{k+1} x^T P_k x^k$ converges absolutely and uniformly on compact subsets $\sigma \bar{\Re}^*$ of $\Re^*$ for $\sigma \in (0,1)$. The gradient $V_x^*(x)=\sum_{k\geq 1} P_k x^k$ of $V^*$ is absolutely convergent for each $x\in \Re^*$ since $r^*_\upsilon \leq r_{\rm con}(V_{x_i}^*(z\upsilon))$, $i=1,..,n$, where $V_{x_i}^*=\partial V^*/\partial x_i$. This proves that $V^*(x)$ is at least continuous at each $x \in  \Re^*$ hence uniformly bounded on ${\bar{\sigma}\bar{\Re}^*}$ for $\sigma< \bar{\sigma} < 1$. Thus $\sup_{x\in \bar{\sigma} \bar{\Re}^*} |V^*(x)| \leq \beta$ for some finite $\beta>0$. Furthermore, by the Cauchy's estimate formula \citep[Theorem 1.3.3]{yoshida1989quadratic,scheidemann2005introduction}, the coefficient of the $(k+1)$-th order  of the Taylor series expansion of $V^*(z\upsilon)$, $|z|<{\bar{r}_\upsilon}$, satisfies $\frac{1}{k+1}|\upsilon^TP_k\upsilon^k| \leq \beta \bar{r}_\upsilon^{-k-1}$, $\forall \upsilon \in \mathbb{R}^n$, $\|\upsilon\|=1$, where $\bar{r}_\upsilon=\bar{\sigma} r^*_\upsilon$. Letting $r_\upsilon :=\sigma r^*_\upsilon < \bar{r}_\upsilon$ and evaluating the sum of the absolute values of the Taylor series for $V^*$, we have
\begin{align*}
    \sup_{x\in {\sigma \bar{\Re}^*}} \sum_{k\geq 1} \frac{|x^T P_k x^k|}{k+1} & = \sup_{\|\upsilon\|=1}\sum_{k\geq 1} \sup_{0 \leq r \leq r_\upsilon} \frac{|\upsilon^T P_k \upsilon^k|}{k+1}r^{k+1} \\
    & \leq \beta \sum_{k\geq 2}\Big(\frac{r_\upsilon}{\bar{r}_\upsilon}\Big)^k \leq \frac{\sigma}{\bar{\sigma}}\frac{\beta\sigma}{\bar{\sigma}-\sigma} < \infty
\end{align*}
which proves the uniform and absolute convergence for the Taylor series of $V^*$ and completes the proof.
\end{proof}
Theorem \ref{theorem.conv} can be used to construct an accurate estimate of the ROC of the value function $V^*$, and consequently, the generated optimal controller $u^*$. It is also possible to find an $n$-sphere approximation of the ROC$(V^*)$. From Theorem \ref{theorem.conv}, it can be see that
$r^*:=\Big(\lim_{k \rightarrow \infty}\sup \|P_k\|^{\frac{1}{k}}\Big)^{-1}$ serves as a lower bound for the radius of convergence of $V^*$. Both ROC$(V^*)$ and $r^*$ will be numerically estimated in the next section.
\section{Algorithm Implementation Examples and Simulation Studies} \label{sec_Algorithm Implementation Examples and Simulation Studies}

In this section, we implement the proposed algorithm on control affine nonlinear systems. A MATLAB routine that takes a symbolic system’s dynamics vector $f(x)$, a symbolic input matrix $g(x)$, a penalizing positive definite function $Q(x)$, an input penalizing positive definite function $R(u)$, a desired expansion order of the system's dynamics, and a desired order of approximation of the optimal control was developed by the authors applying the described algorithms in Section \ref{sec_Infinite Horizon Nonlinear Regulator for Control Affine Systems}. The MATLAB routine efficiently computes the solution matrices of the Taylor expanded value function, named $P_k$'s in the above algorithm, up to the prescribed order of approximation (available upon request). Using the MATLAB routine, we are able to compute very high truncation orders very fast (e.g. the $300^{th}$ order value function for second order systems). Three examples with different systems natures and nonlinearities are presented including an input constrained problem. 

\subsection{Third Order Multi-input System}
In this example, using the proposed HJB approach based nonlinear regulator, we construct the value function and the optimal feedback control law up to  different truncation orders. 
The problem is to solve the optimal control problem given by
\begin{equation} \label{Ex1_cost}
V= \int_{0}^{\infty} (50(x_1^2+x_2^2+x_3^2)+x_1^4+x_2^4 + x_3^4 +  \frac{1}{2} (u_1^2 + u_2^2)) dt
\end{equation}
governed by the dynamics,
\begin{align}\label{Ex1}
\begin{split}
&\dot{x_1}=3 \text{sin}(x_2) \\
&\dot{x_2}= 2 x_1^3+ x_3 + u_1 \\ 
&\dot{x_3}= 3 e^{x_1} - u_2
\end{split}
\end{align}
For this example, the desired order of the regulator is selected to be $30$  so the order of the value function is $31$. Then, $f(x)$, $g(x)$ and $Q(x)$ are expanded to get $F_k$'s, $G_k$'s and $Q_k$'s for $k=1,\dots, 30$. The Taylor series matrix coefficients of the optimal solution, $P_k$'s, were computed in $0.6643$ seconds using a laptop with an Intel(R) Core(TM) i7-8550U CPU at 1.8 GHz and 16GB RAM running on Windows 10. For $\bar{k}=10$ or lower, the built routine is able to instantly produce the matrix coefficients (i.e., in zero MATLAB CPU time). Efficiently computing high powers $P_k$'s helps in estimating the ROC, Fig. \ref{ROC_fine}, by shooting unit vectors in all directions and computing the radius of convergence at each direction as discussed in Section \ref{sec_Convergence Analysis and Radius Estimations}. Now, let us examine the initial condition $x(0)^T= \begin{bmatrix} x_{1}(0) && x_{2}(0) && x_{3}(0) \end{bmatrix} = \begin{bmatrix} -2 && -1.5 && 0\end{bmatrix}$, which is inside the ROC. The closed loop system's response and the control action are shown in Fig. \ref{FirstExample_states} and Fig. \ref{FirstExample_inputs}. 

    \begin{figure}[htb]
        \centering
    \framebox{\parbox{3.2in}  {\includegraphics[width=3.3in]{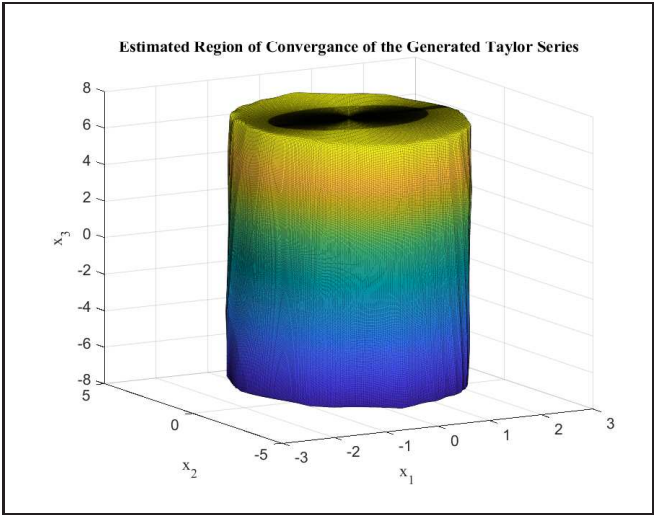}}}
      \caption{Estimated ROC for the Taylor series of the optimal solution of \eqref{Ex1_cost}-\eqref{Ex1}}
      \label{ROC_fine}
   \end{figure}

\begin{figure}[htb] 
\centering
  \begin{tabular}{|c|c|}    
  \hline
    \includegraphics[width=3.3in]{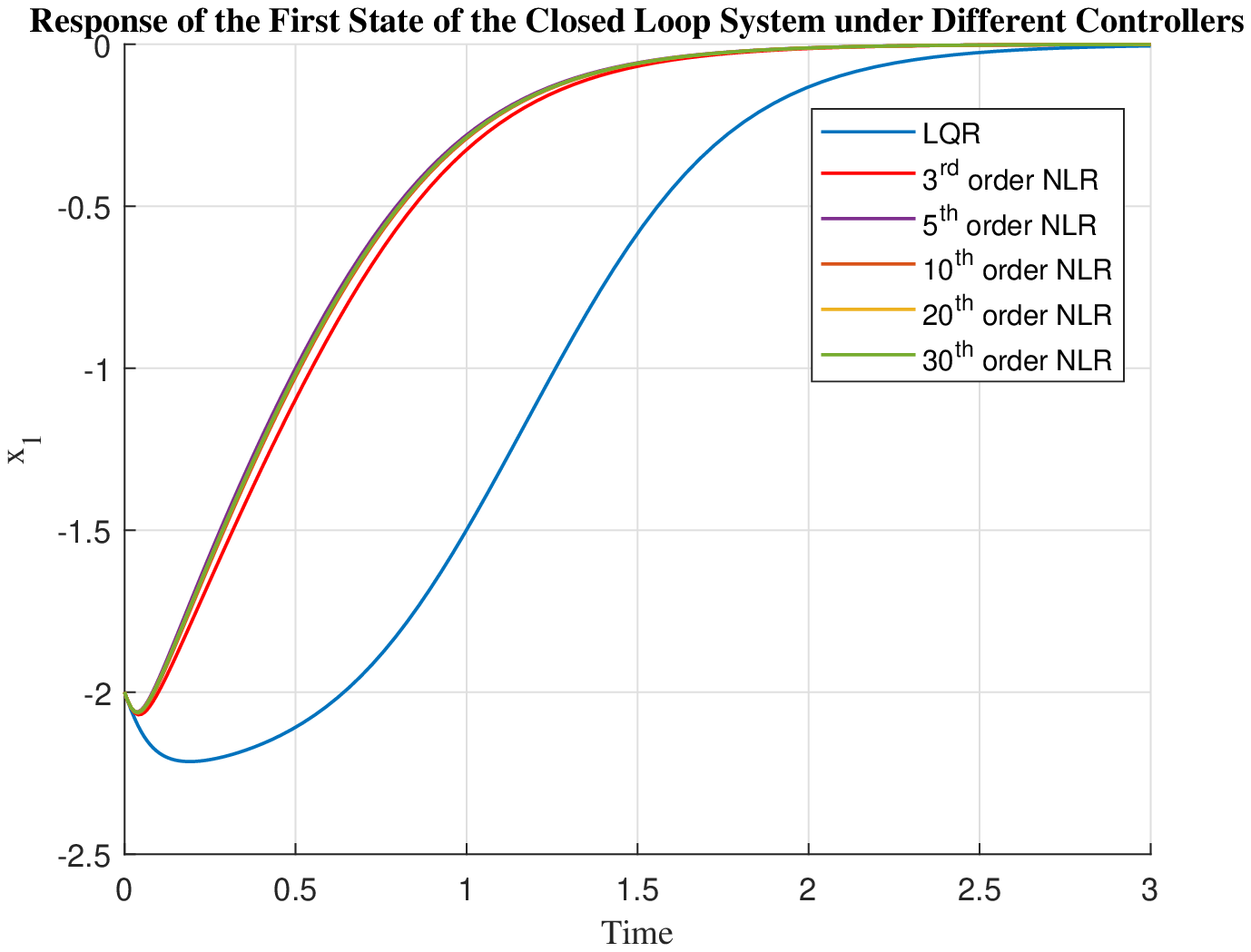} \\ 
    \includegraphics[width=3.3in]{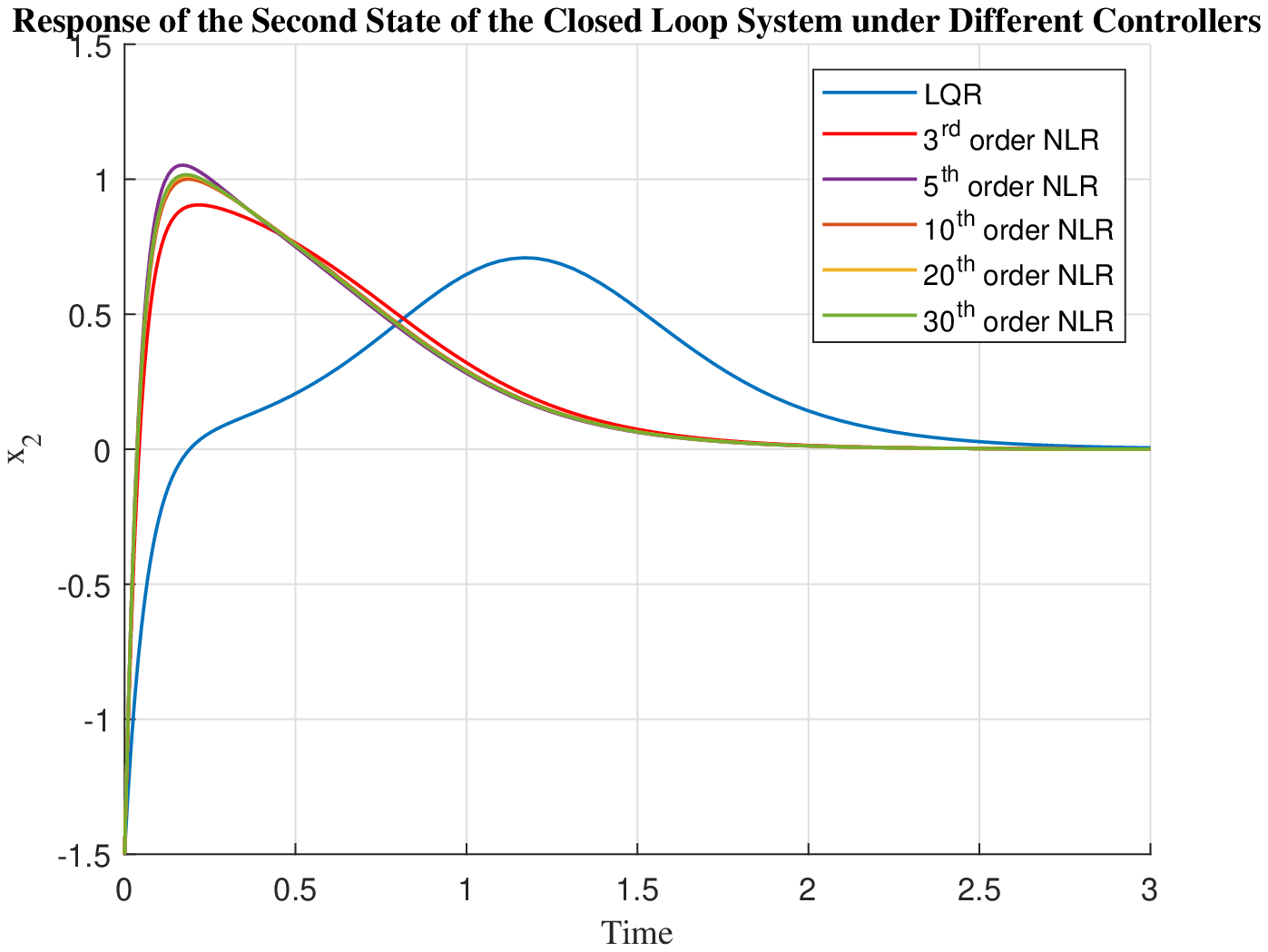}  \\
    \includegraphics[width=3.3in]{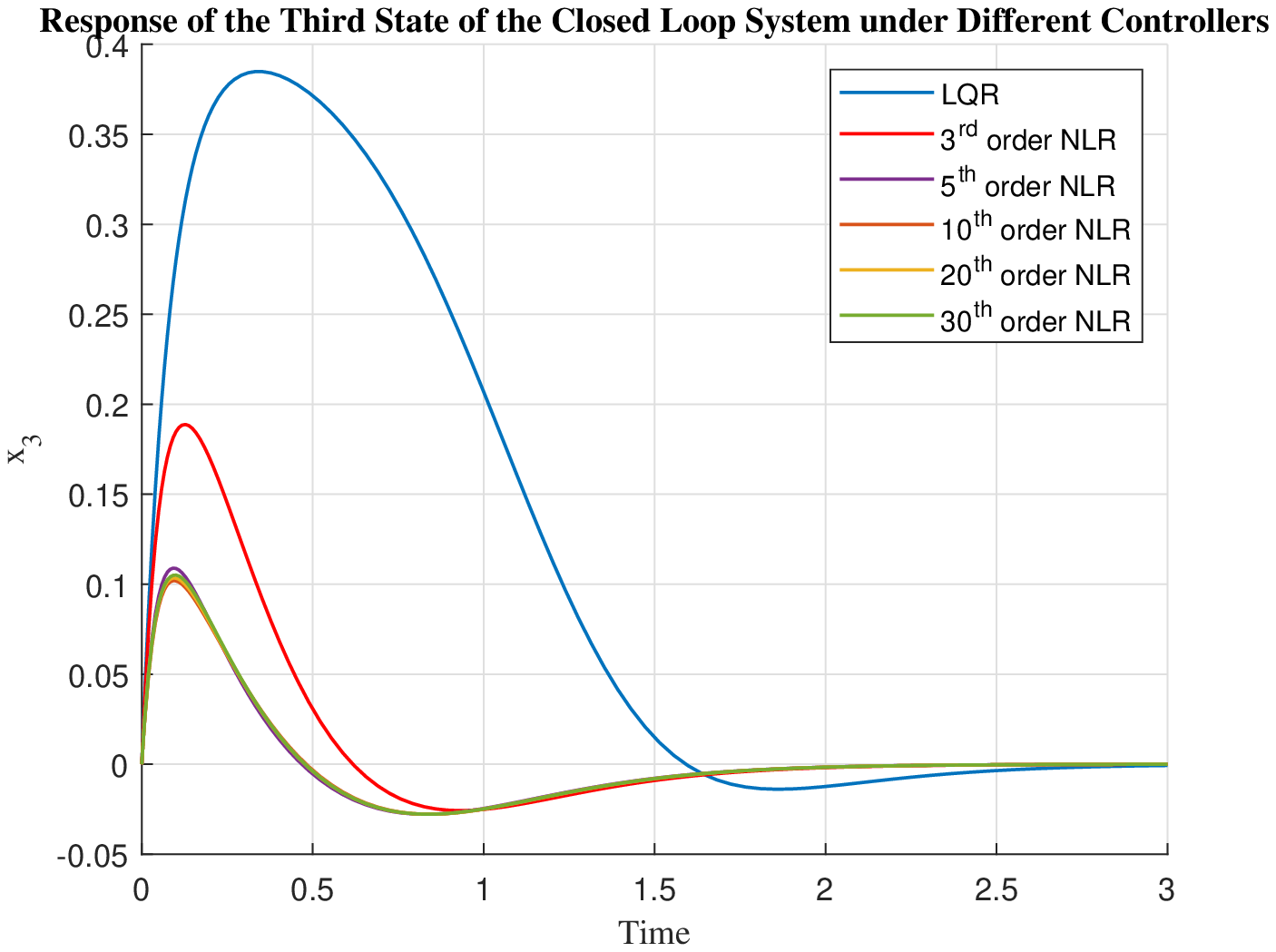}  \\
    \hline
  \end{tabular} 
   \caption[]{Comparing the closed loop response of the nonlinear system under the control of the LQR and different orders of the NLR with an initial condition of $x(0)^T= \begin{bmatrix}-2&&-1.5&&0\end{bmatrix}$.} 
   \label{FirstExample_states}
\end{figure}

  
   
   
   \begin{figure}[htb] 
\centering
  \begin{tabular}{|c|c|}    
  \hline
    \includegraphics[width=3.3in]{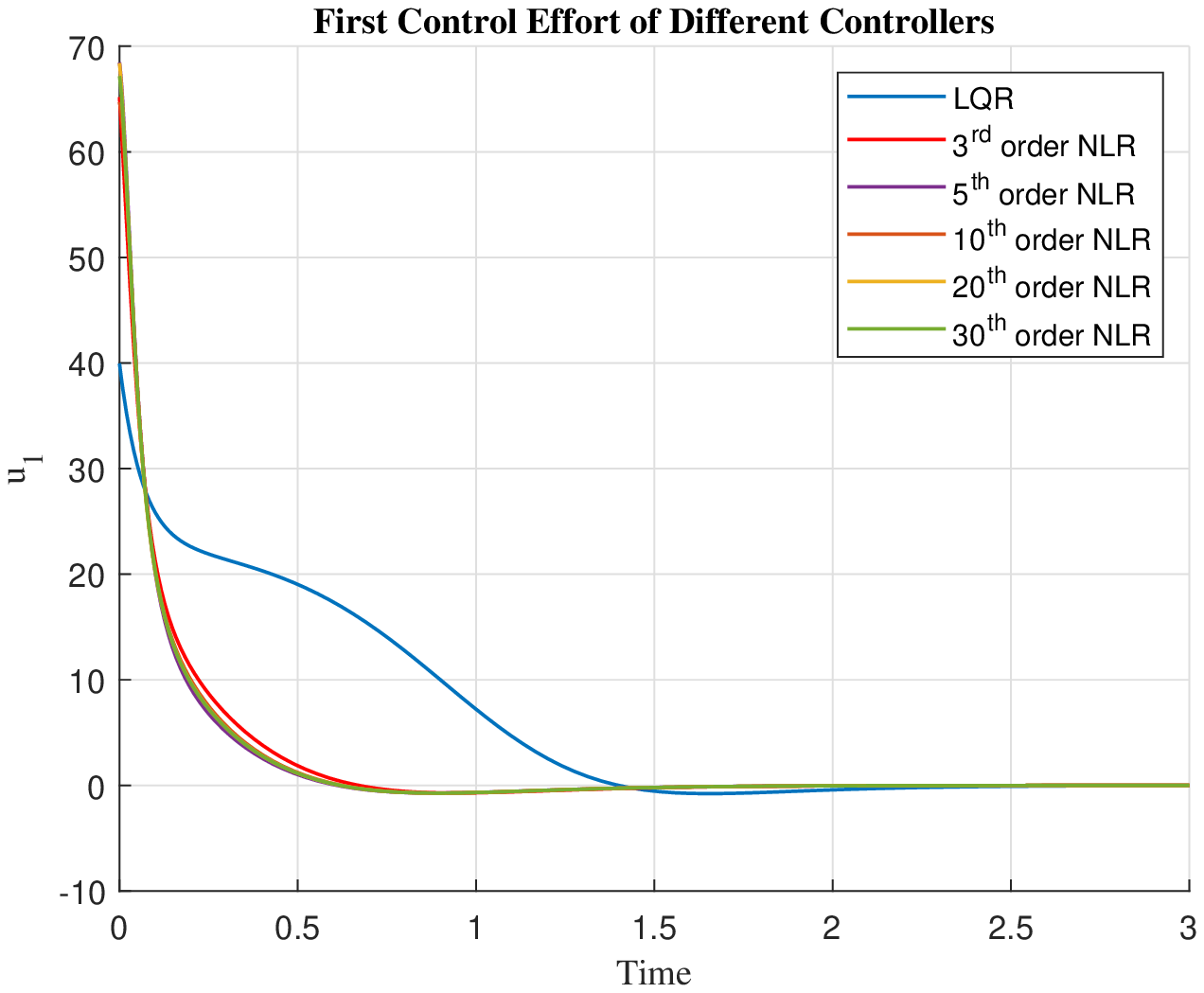} \\ 
    \includegraphics[width=3.3in]{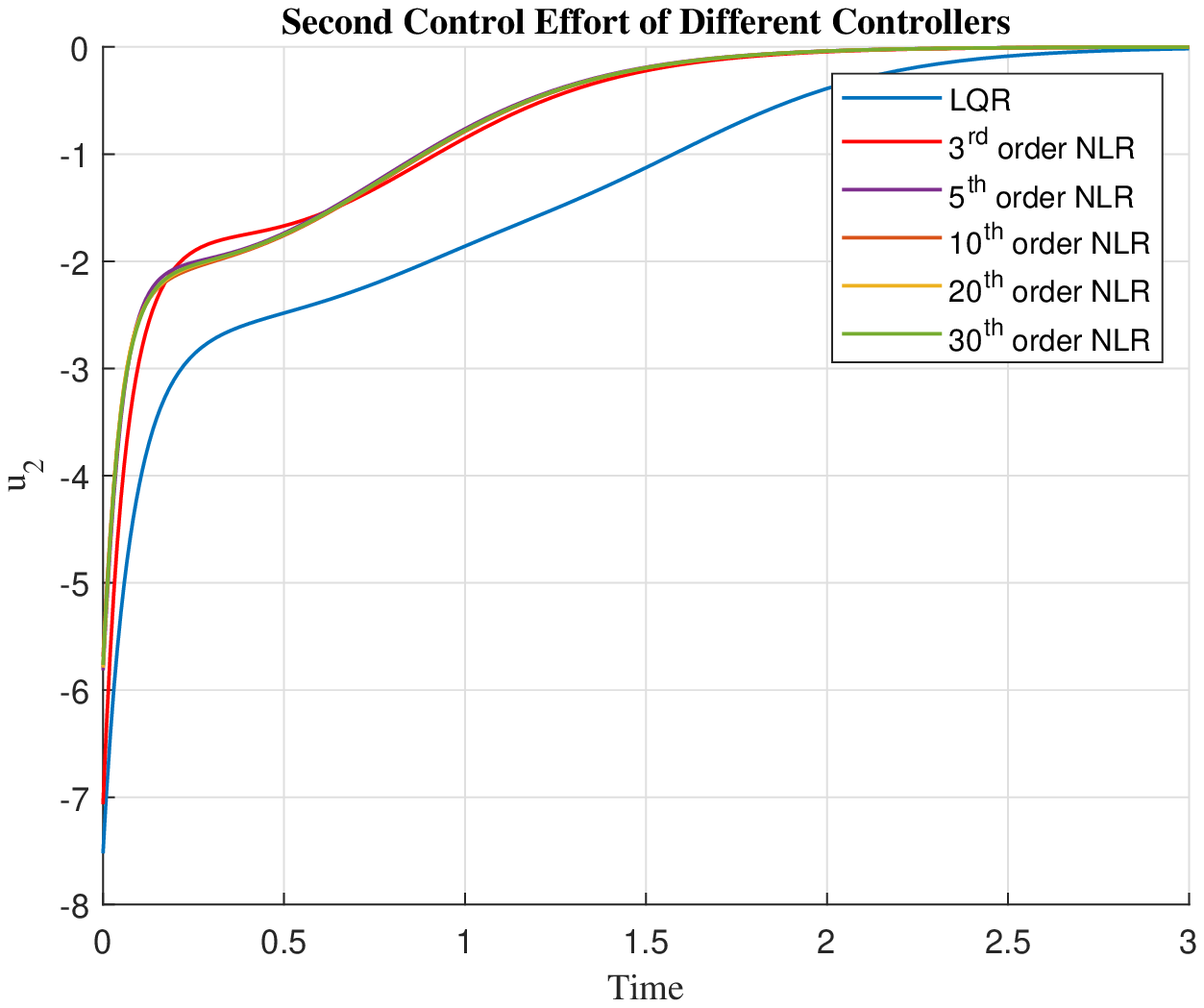}  \\
    \hline
  \end{tabular} 
   \caption[]{First and second control actions of the LQR and different orders of NLR to stabilize the system \eqref{Ex1} with an initial condition of $x(0)^T= \begin{bmatrix}-2&&-1.5&&0\end{bmatrix}$.}
   \label{FirstExample_inputs}
\end{figure}
Clearly, higher NLR's provide better performances. It can be seen, however, that the improvement in the performance after the $5^{th}$ NLR is not notable, which can be seen as a sign of convergence to the optimal controller. We tested more initial conditions and there was no significant improvement in the performance after the $7^{th}$ NLR in most cases. Moreover, if we slightly increase this initial condition, the LQR, is not capable of stabilizing the system unlike the NLR's and thus the NLR's result in larger region of attraction. It is worth mentioning, however, that if we start outside the ROC, the nonlinear controllers, consume higher control power in the beginning to stabilize the system in a faster pace which may not be always realizable. Nonetheless, one could impose some inputs constraints and the NLR's are capable of regulating the system in many cases, as we will show for the next example. Additionally, we must note that we have tested more initial conditions and it was concluded that if we are outside the region of convergence of the Taylor series, very high order regulators are not guaranteed to provide better performance neither stability, as expected. Thus, the control designer may choose relatively lower order NLR's, e.g. $5^{\text{th}}$ to the $10^{\text{th}}$, for this problem as they provide a greater performance with reasonable actuation. 
\subsection{F-8 Flight Control System} 
This system is taken from \cite{garrard1977design} which they applied their HJB equation based control. \cite{beeler2000feedback} used the control-affine, constant-input-matrix version of this example to compare handful of feedback control methodologies. Garrad and Jordan's controller, called the two-term Taylor expansion method in \cite{beeler2000feedback}, was compared with other controllers including a state-dependent riccati equation (SDRE) controller. \cite{AlmubarakSadeghandTaylor2019} used that version too to implement their NLQR.

This F-8 flight control system is given by 
$$f(x,u)=\begin{bmatrix}
-0.877x_1+x_3+0.47x_1^2-0.088x_1x_3 \\ -0.019x_2^2 +3.846x_1^3-x_1^2x_3 -0.215u\\+0.28 x_1^2 u+0.47x_1 u^2+0.63u^3 \\
x_3 
\\
-4.208x_1-0.396x_3-0.47x_1^2-3.564x_1^3 \\+20.967 u +6.265x_1^2 u+46 x_1 u^2 +61.4u^3 
\end{bmatrix}_{3 \times 1}$$
where $x_1$ is the angle of attack deviation (rad) from the trim value of $0.044$, $x_2$ is the flight path angle (rad), $x_3$ is the rate of change in the flight path angle (rad/sec) and $u$, the input, is the tail deflection deviation (rad) from the trim value of $-0.009$. The cost functional, as in \cite{garrard1977design}, is chosen to be
$$ V=\frac{1}{2} \int_{0}^{\infty} (x^T 0.25 I_{3 \times 3} x + u^2) dt$$
To use our proposed algorithm, we first need to put the system of the form of \eqref{gen.sys}, which is a control affine dynamical system. 
We want to develop a nonlinear feedback control that is able to regulate the angle of attack optimally. In \cite{beeler2000feedback}, Beeler, Tran and Banks used an initial condition of $x(0)=\begin{bmatrix}0.4363&&0&&0\end{bmatrix}^T$, i.e. an angle of attack of $25^o$, to compare the different feedback methodologies. It is reported that Garrard's algorithm is simple and very effective but only for systems with low nonlinearites since it is not feasible to get higher approximation orders of the optimal control. However, the SDRE and Garrard's controllers were outweighed by the other methods. Almubarak, Sadegh and Taylor in \cite{AlmubarakSadeghandTaylor2019} obtained higher orders of approximation of the optimal control. Nevertheless, the developed controllers in these papers used the approximated system, with a constant input matrix, as we mentioned. Here we use the general control affine system, i.e. with a state dependent input matrix $g(x)$, which surely gives a better solution and more accurate approximation to the optimal solution. Using our proposed algorithm, we got up to the $30^{th}$ order of approximation more accurately. The estimated ROC for this problem is shown in Fig. \ref{ROC_flight}. The estimate of $r_{\rm con}(V^*)$ is equal to $r^*=0.52$, which fairly well agrees with the ROC range for $x_1$. We provide here the $5^{th}$ order controller, after removing very small and zero terms without affecting the performance, which is an improved version of the NLQR provided in \cite{AlmubarakSadeghandTaylor2019}:
\begin{align*}
u=&-0.053x_1+0.5x_2+0.521x_3+ 0.035x_1^2  -0.045x_1 x_2 \\ &
+0.339x_1^3-0.531x_1^2 x_2+0.017x_1^2 x_3+0.139 x_1 x_2^2 \\& 
-0.042x_1 x_2 x_3 +0.013 x_1 x_3^2 +0.504x_1^4-0.655x_1^3 x_2\\ &
+0.082 x_1^3 x_3+0.353x_1^2 x_2^2-0.081  x_1^2 x_2 x_3-0.087 x_1 x_2^3\\ &
+0.0327 x_1 x_2^2 x_3 + 2.29 x_1^5 -3.205x_1^4 x_2+0.499x_1^4 x_3\\ &
+2.104 x_1^3 x_2^2 -0.554 x_1^3 x_2 x_3 +0.043 x_1^3 x_2^2 -0.864  x_1^2 x_2^3 \\ &
+ 0.271  x_1^2 x_2^2 x_3 -0.038 x_1^2 x_2 x_3^2+0.155  x_1 x_2^4  \\ &
-0.087x_1 x_2^3 x_3+0.011 x_1 x_2^2 x_3^2+0.013 x_2^4 x_3
\end{align*}
As shown in Fig. \ref{FlightDiffNLQR25}, the performance gets improved as we use higher powers. Yet, clearly after the $10^{th}$ order approximation, the performance almost did not get enhanced. These results are an improvement of the results obtained in \cite{AlmubarakSadeghandTaylor2019} and is very close, if not better than, to the best results obtained by \cite{beeler2000feedback} through the interpolation of two-point boundary-value (TPBV) open-loop control.

          \begin{figure}[htb]
        \centering
    \framebox{\parbox{3.2in}  {\includegraphics[width=3.3in]{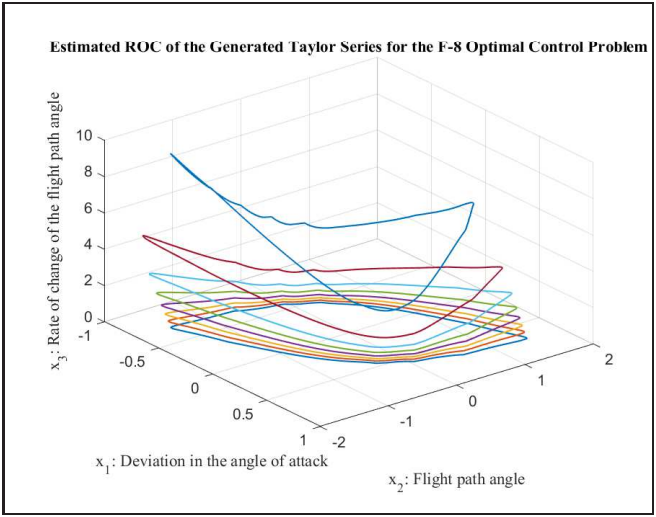}}}
      \caption{Estimated ROC boundary contours for constant inclination angle for the generated Taylor series}
      \label{ROC_flight}
   \end{figure}

\begin{figure}[htb] 
\centering
  \begin{tabular}{|c|c|}    
  \hline
    \includegraphics[width=3.3in]{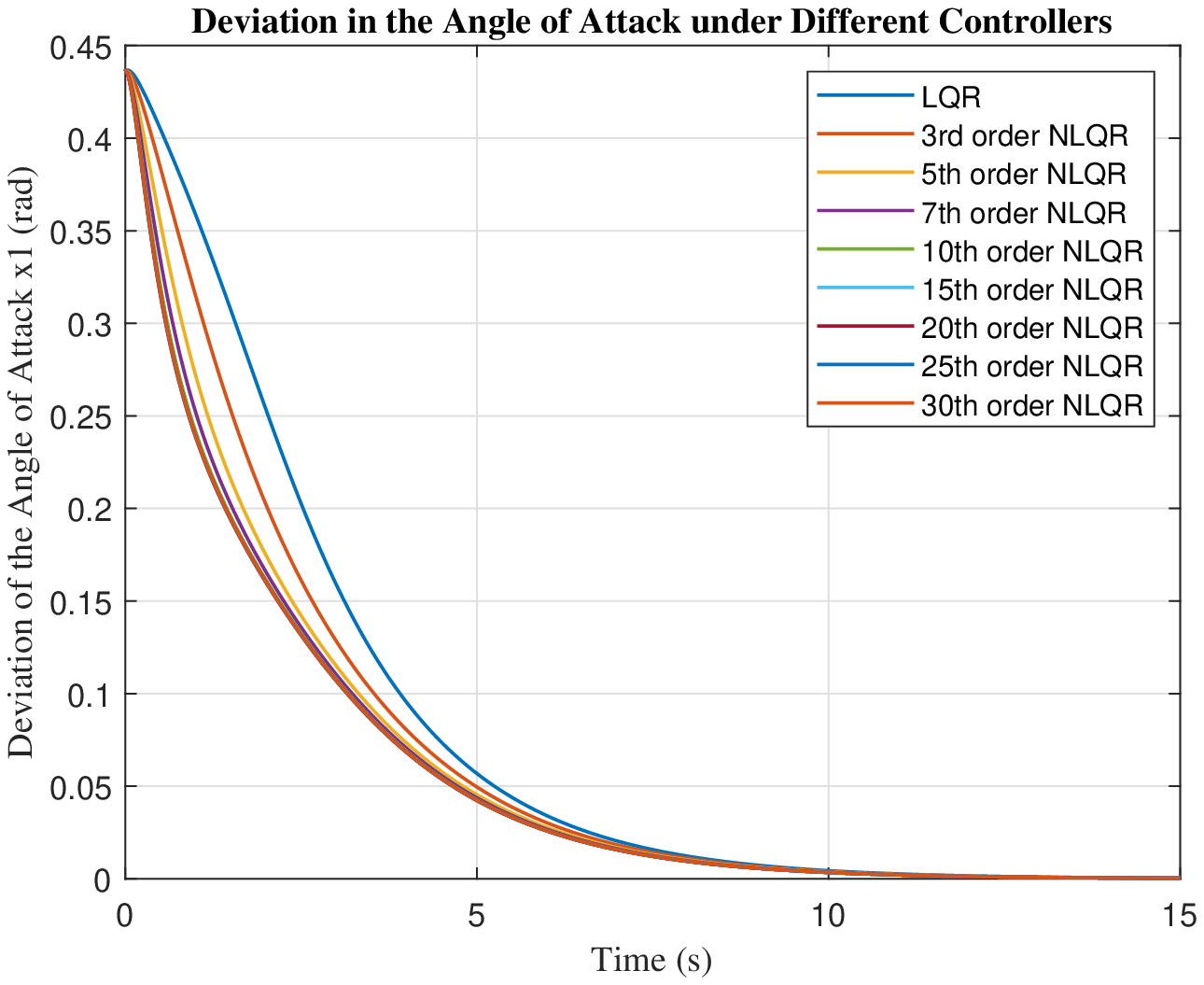} \\ 
    \includegraphics[width=3.3in]{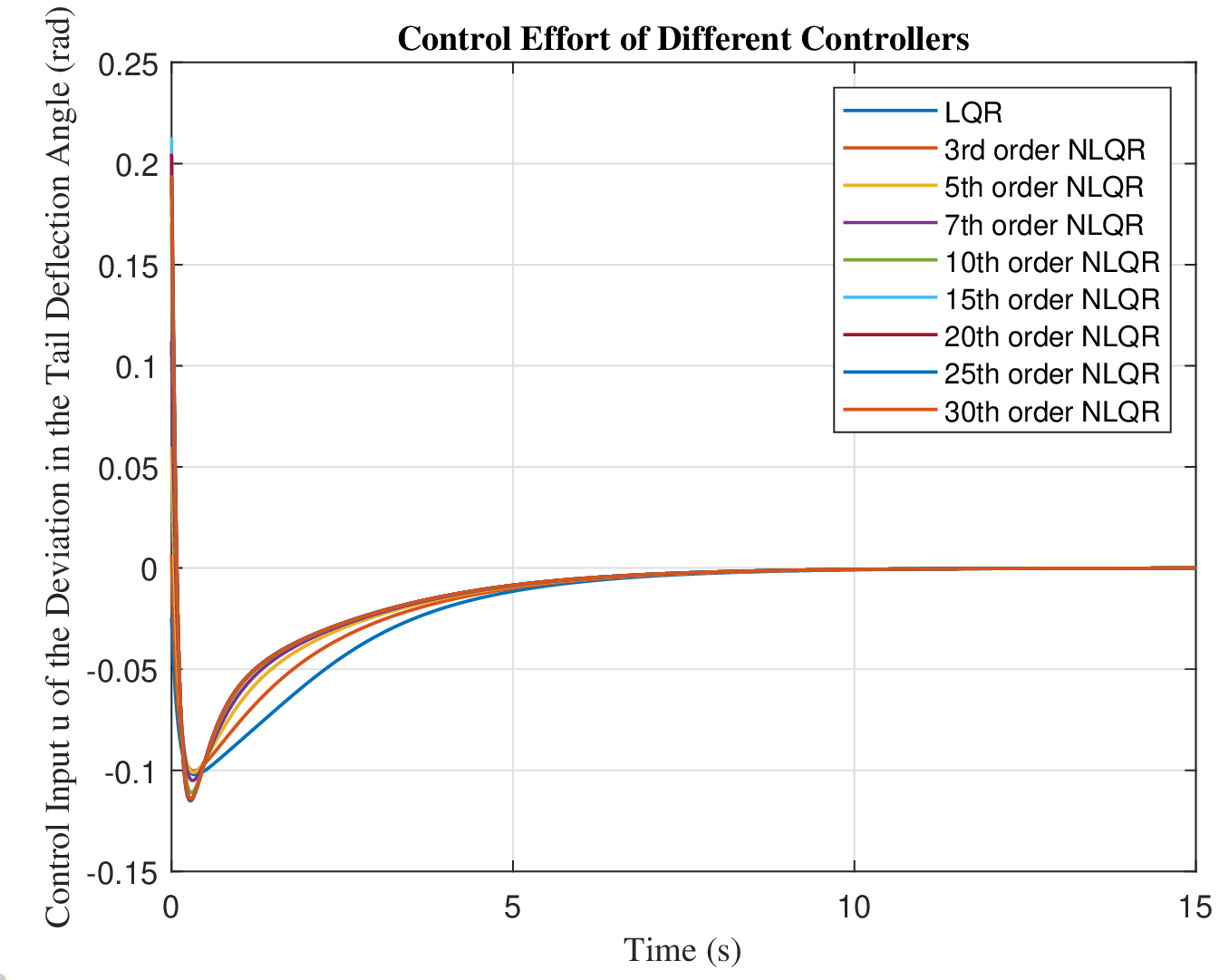}  \\
    \hline
  \end{tabular} 
   \caption{The angle of attack's deviation and the required deviation of the tail deflection, i.e. control action, of different Taylor expansions of the optimal control to regulate the angle of attack with an initial angle of \(25^o\).}
   \label{FlightDiffNLQR25}
\end{figure}
   
   
It is worth mentioning that for higher angles of attack, low order controllers, including the LQR obviously, are not capable of regulating the system. For a detailed discussion about the F-8 flight control system, the reader may refer to \cite{garrard1977design}. Fig. \ref{FlightDiffNLQR_30} shows the performance of higher order controllers, $6^{th}$, $7^{th}$ and $30^{th}$, when the initial condition is $x(0)=\begin{bmatrix}0.5236&&0&&0\end{bmatrix}^T$, which corresponds to an angle of attack of $30^o$ where low orders are not capable of regulating the system. Clearly, there is a significant improvement in the performance when using the  $30^{th}$ power control rather than the  $6^{th}$ power control. This initial condition, however, is outside the ROC and thus, as discussed before, using higher orders does not guarantee stability nor better performance. In fact, many of the higher order regulators could not stabilize the system. 
\begin{figure}[htb]
\centering
  \begin{tabular}{|c|c|}    
  \hline
    \includegraphics[width=3.3in]{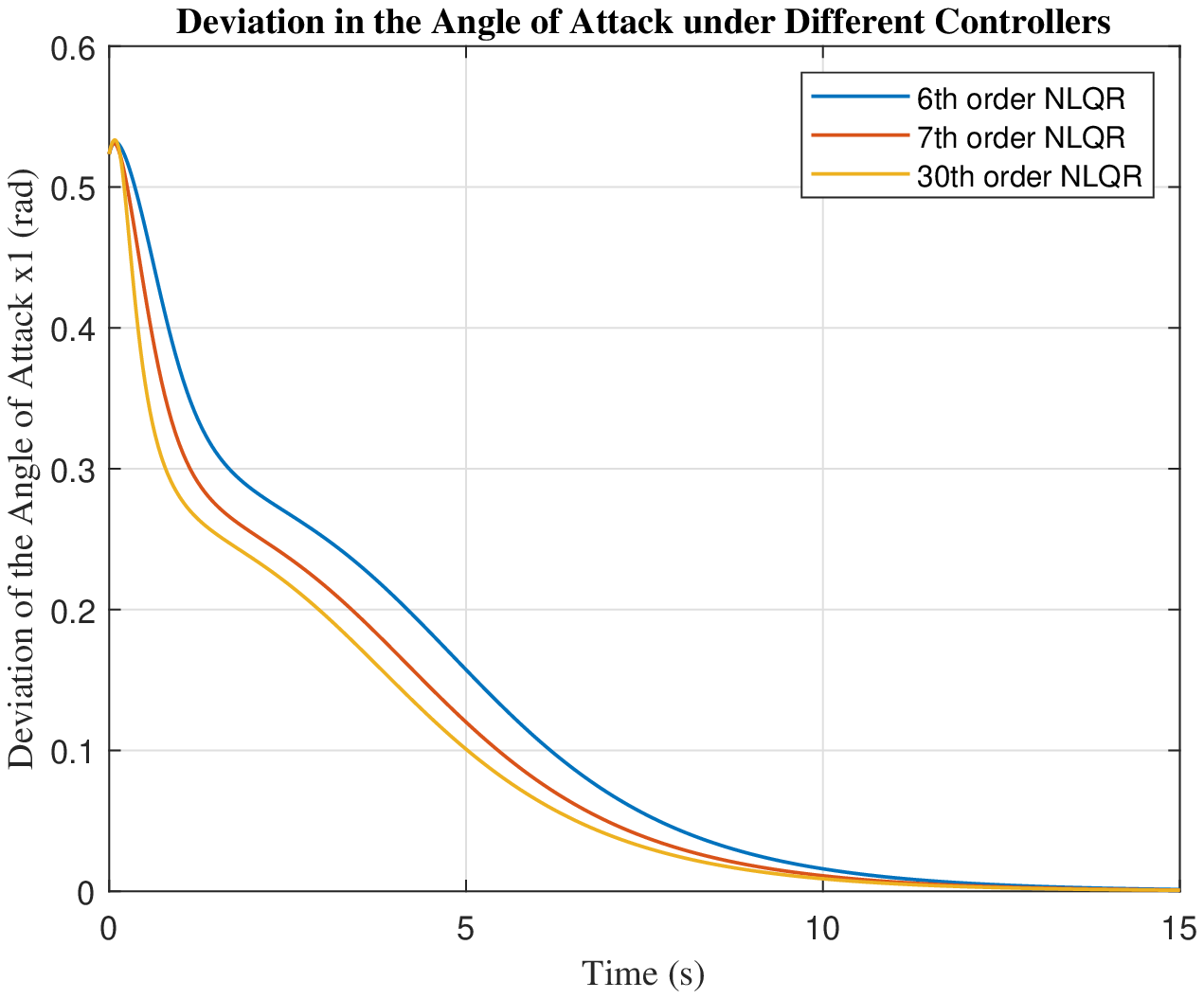} \\ 
    \includegraphics[width=3.3in]{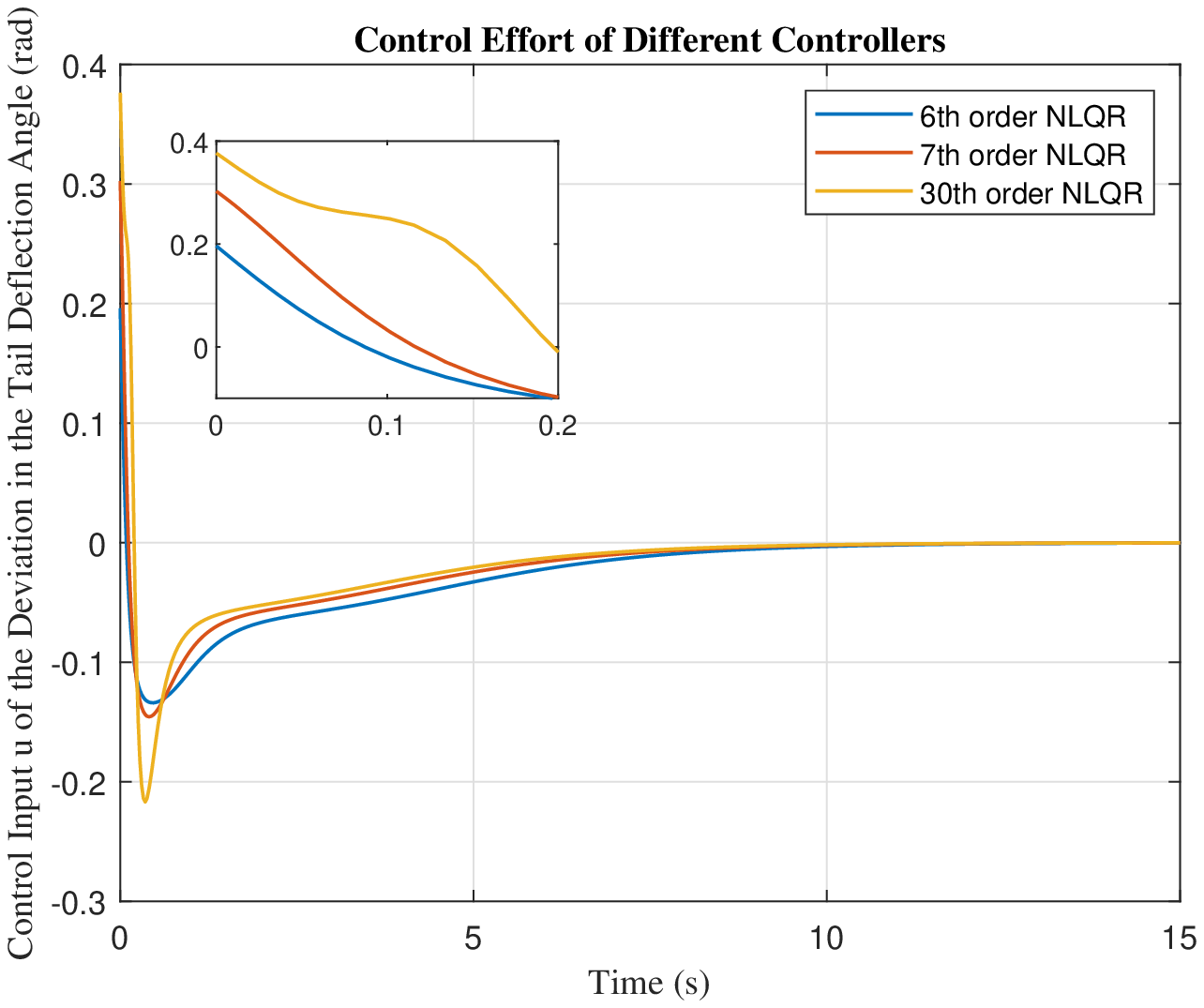}  \\
    \hline
  \end{tabular} 
   \caption{The angle of attack's deviation under the control of $6^{th}$, $7^{th}$ and $30^{th}$ Taylor expansion approximations of the optimal control with an initial angle of \(30^o\) where using low orders result in unstable closed loop system.}
    \label{FlightDiffNLQR_30}
\end{figure}

   

\subsection{Constrained F-8 Flight Control System} 
We impose input constraints, which is reflected in the cost integral as in \cite{abu2005nearly} where artificial neural networks were used to approximate the optimal solution. Fig. \ref{FlightDiffNLQR_30} shows that the nonlinear controller uses high input power to recover and stabilize the angle of attack. Testing this system with multiple nonlinear regulators showed that the NLR's usually can't recover if the deviation of the tail deflection required is very high. Thus, we restrict the agnle of deviation to $0.2$ rad, i.e. $ -11.46^o \leq u \leq 11.46^o $. Then, to confront this constraint, $\phi$ is chosen to be $\phi(v)= \tanh(5v)/5$, and thus $\Psi(v)=\ln(\cosh(5v))/25$ and consequently the cost functional will be
$$
V=\int_{0}^{\infty} \Big( \frac{1}{2} x^T 0.25 I_{3 \times 3} x + \frac{1}{5}\int_{0}^{u}\tanh^{-1}(5v) dv \Big) dt
$$
Therefore, the $\bar{k}^\text{th}$ NLR will be in the form $u_{\bar{k}}(x)=-\frac{1}{5}\tanh(5 g^T \sum_{k=1}^{\bar{k}} P_k x^k$). The results in Fig. \ref{FlightDiffNLQR_30_Sat} show a great deal in handling the constraints and a clear improvement in the NLR to recovery. Notice that when the constraints are lifted as in Fig. \ref{FlightDiffNLQR_30}, the $7^\text{th}$ order NLR needed to use a relatively high deviation in the tail to recover while after imposing the constraints, it was able to regulate the system using less control power. Moreover, some NLR's couldn't stabilize the system before adding the constraints but imposing input saturation helped in generating NLR's that are capable of handling the high angle of attack. The results here show how to untangle input saturation directly by incorporating a saturation function in the cost integral without the need of adding more states to enforce saturation indirectly as in \cite{AlmubarakSadeghandTaylor2019}.

\begin{figure}[htb] 
\centering
  \begin{tabular}{|c|c|}    
  \hline
    \includegraphics[width=3.3in]{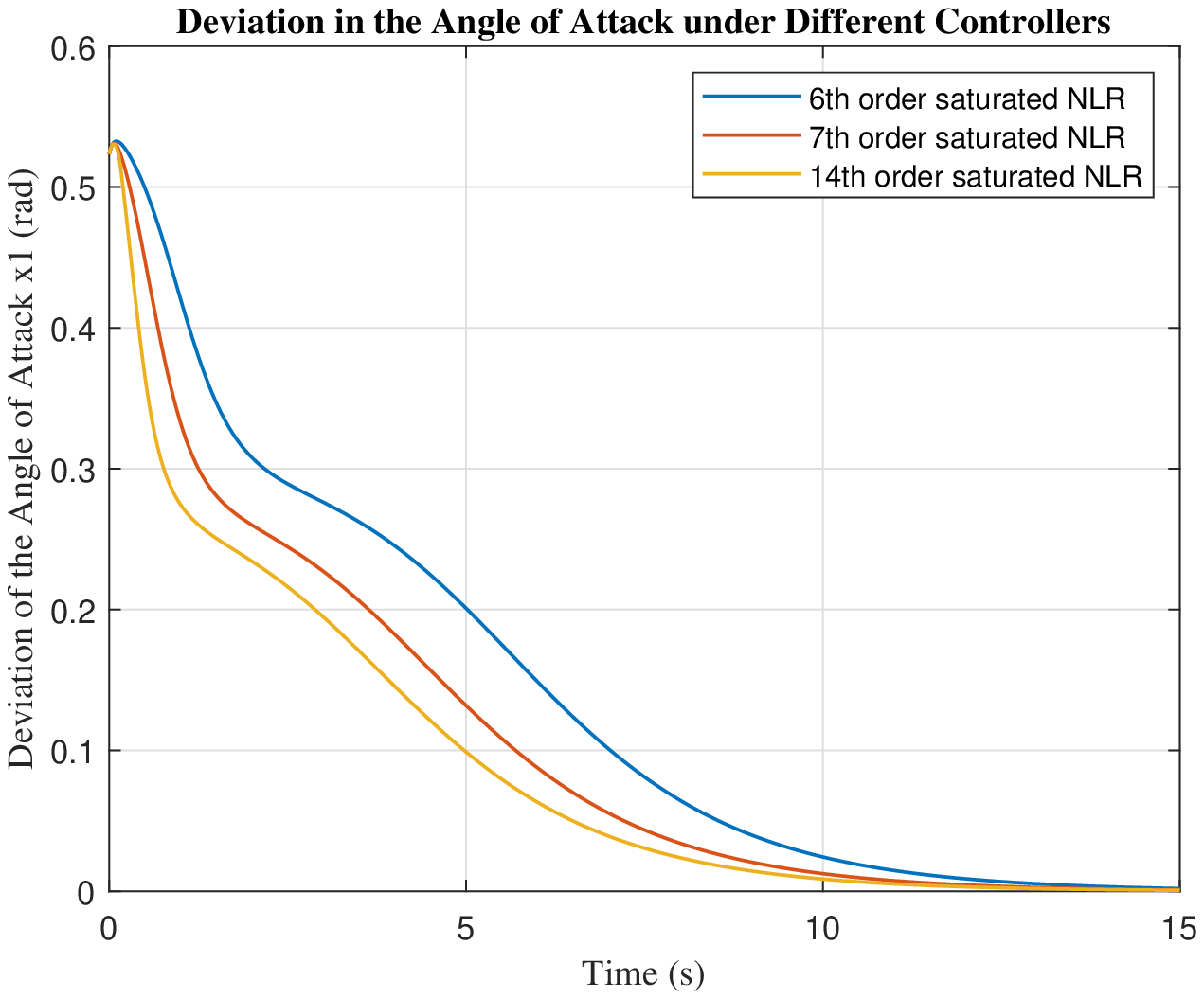} \\ 
    \includegraphics[width=3.3in]{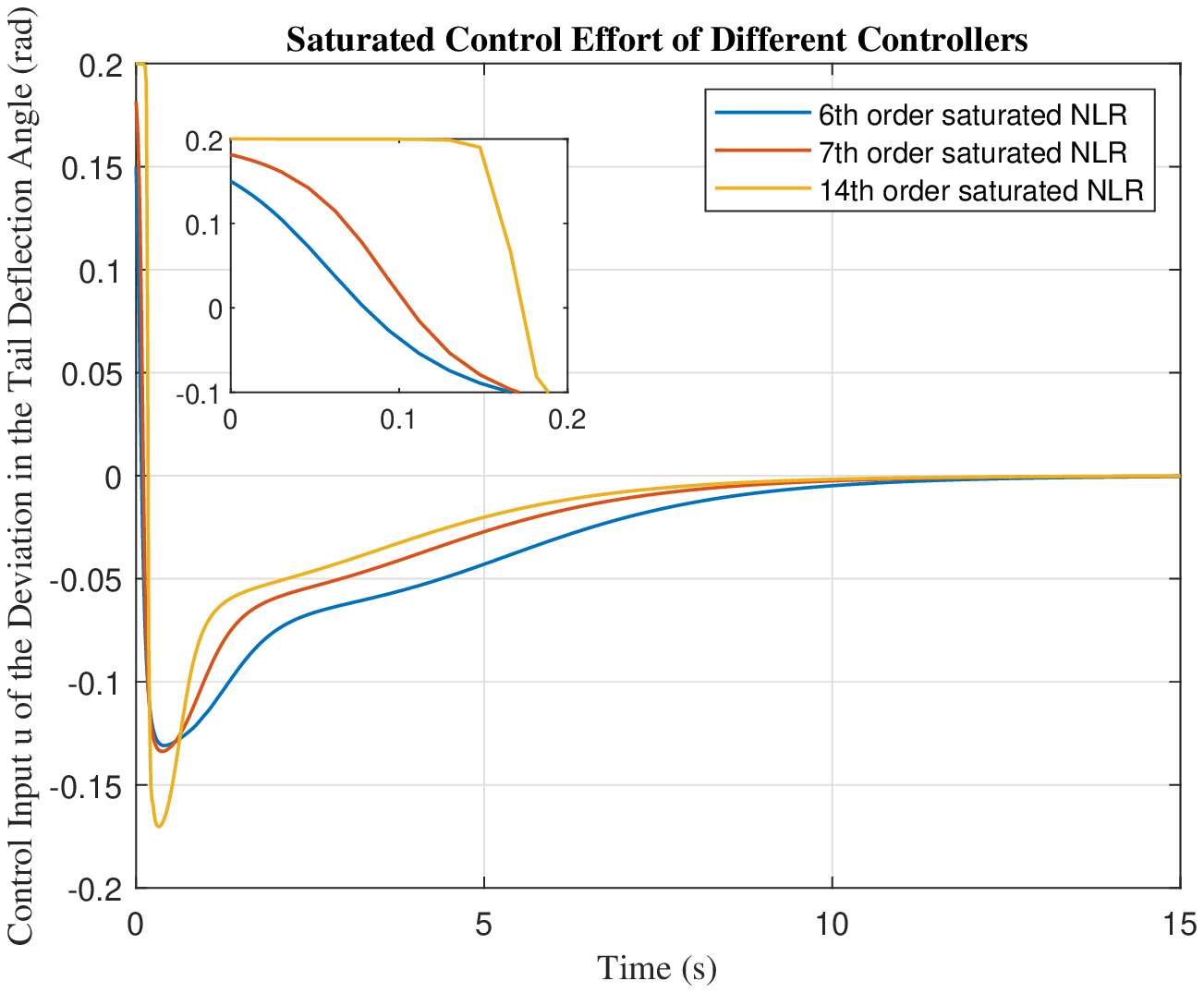}  \\
    \hline
  \end{tabular} 
   \caption{The angle of attack's deviation under the constrained control of the $6^{th}$, $7^{th}$ and $30^{th}$ NLR's with an initial angle of attack of \(30^o\) where using low orders result in unstable closed loop system.}
   \label{FlightDiffNLQR_30_Sat}
\end{figure}



\section{Conclusion} \label{sec_Conclusion}
The paper presented an efficient and novel algorithm to compute the value function and the associated optimal controller for control-affine nonlinear systems. The algorithm can also be utilized to accurately estimate the ROC of the generated optimal controller. More specifically, a general closed form solution for each matrix coefficient of the value function and the resulting optimal controller was provided. The methodology is based on efficiently expanding the HJB equation to construct a nonlinear matrix equation that can be untangled independently of the current states. As a result, the NLR can be obtained offline using a minimal polynomial basis function that includes all possible combinations of the states as a generalization of the linear case. It was demonstrated that the proposed methodology is capable of achieving asymptotic stability for nonlinear systems satisfying the required assumptions. Three examples of nonlinear systems were presented including one with input saturation where it was shown that using higher order controllers improves performance inside the region of convergence of the series and could provide larger region of attraction.

Future works may include performing the power series expansion of the value function around arbitrary states in order to possibly expand the region of convergence. Another improvement could be extending the current work to general nonlinear systems, i.e., not affine in control without adding more state variables. Furthermore, optimal estimation algorithms could be formulated similarly to the proposed method.


\appendixtitleon
\appendixtitletocon
\begin{appendices}  
    
    \section{Proof of Proposition \ref{psi and phi prop}} \label{psi and phi prop proof}
    \begin{proof}
We claim that $\Psi(v)= v^T \phi(v) - R \circ \phi (v)$. Differentiating both sides with respect to $v$ yields
\begin{equation}
\frac{\partial \Psi}{\partial v} = \phi(v)^T  + v^T \frac{\partial \phi}{\partial v}  - \frac{\partial R}{\partial u} \frac{\partial \phi}{\partial v}  =  \phi(v)^T
\end{equation}
using that ${\partial R (u)}/{\partial u}=\rho(u)^T=v^T$ evaluated at $u=\phi(v)=\rho^{-1}(v)$. The analyticity of $\Psi$ follows from the analyticity of $\phi$ being the inverse of an analytic function with positive definite (and invertible) Jacobian.
\end{proof}

    \section{Proof of Proposition \ref{vectors power prop}} \label{vectors power prop proof}
\begin{proof}
 $\langle x^k, y^k \rangle = (x^k)^T y^k= (x^k L_k )^T L_k y^k= (x^{\otimes k})^T y^{\otimes k} \\ = (x^T y)^k$, where $L_k$ is a unique linear mapping matrix such that $x^{\otimes k}=L_k x^k$. 
\end{proof}
    \section{Proof of Proposition \ref{prop:sym}}
    \label{Sym prop proof}
    \begin{proof}
    $i$) If \eqref{opt.grad_k} holds then the Jacobian of $P_k x^k$, being the Hessian of $V_k$, must be symmetric. To prove the converse suppose that the Jacobian of $P_k x^k$ is symmetric. Then
\begin{align*}
(k+1) \left(  \frac{\partial V_{k}}{\partial x}\right)  ^{T} &  =P_{k}x^{k}+\left(
P_{k}\frac{\partial x^{k}}{\partial x}\right)  ^{T}x= (k+1) P_k x^{k}%
\end{align*}
where the last equation follows from
$$
\frac{\partial x^{k}}{\partial x}x=\sum_{i=1}^{n}\frac{\partial x^{k}%
}{\partial x_{i}}x_{i}=kx^{k}%
$$
$ii$)By Kronecker product properties, the symmetry condition is equivalent to 
\begin{align*}
& \frac{\partial }{\partial x} x^T P_k x^k = (k+1) P_k x^k  \iff \\
& \text{vec}(P_k)^T \Big( \frac{\partial }{\partial x} ( x^k \otimes x) - (k+1) ( x^k \otimes I)  \Big) = 0 
\end{align*}
Denoting the matrix multiplying $\text{vec}(P_k)^T$ by $N(x)$, we have
$N(x) = (\frac{\partial}{\partial x} x^k \otimes x) - k(x^k \otimes I)$. We show that $( y^k \otimes y)^T N(x) =0, \; \forall x,y \in \mathbb{R}^n$:
\begin{gather*}
    \begin{split}
(y^k \otimes y)^T N(x) = & ( y^{kT} \otimes y^T) ( \frac{\partial}{ \partial x} x^k \otimes x) \\ & - k (y^{kT} \otimes y^T)(x^k \otimes I) \\ 
\Rightarrow (y^k \otimes y)^T N(x) & = \big( \frac{\partial}{\partial x} \langle y^k , x^k  \rangle \otimes y^T x \big) \\ & - k ( \langle y^k , x^k \rangle \otimes y^T) 		
    \end{split}
\end{gather*}
By Proposition \ref{vectors power prop}, $\langle y^k , x^k \rangle = (y^T x)^k$ and $	\frac{\partial}{\partial x} (y^T x)^k = k(y^T x)^{k-1} y^T$. It follows that $(y^k \otimes y)^T N(x)=0$ thus
proving the claim. Now,
\begin{gather*}
(y^k \otimes y)^T N(x) = y^{(k+1)T} K_k N(x) =0, \;\forall y \in \mathbb{R}^n \\
\Rightarrow \mathcal{R}(N(x)) \subseteq \mathcal{N}(K_k)
\end{gather*}
Now, if $\text{dim}\big( \oplus_{x \in \mathbb{R}^n} \mathcal{R}(N(x)) \big) < \text{dim}( \mathcal{N}(K_{k})$, then $\exists z\in \mathcal{N}(K_k)$ such that $z^T N(x) =0$,  $\forall x \in \mathbb{R}^n$. Next, let a matrix $Z$ be such that $\text{vec}(Z) = z$. Then, $z^T N(x)=0$ implies
\begin{align*}
\frac{z^T}{k+1} \big(\frac{\partial }{\partial x} ( x^k \otimes x) - (k+1) ( x^k \otimes I) \big) = \\
\frac{1}{k+1} \frac{\partial}{\partial x} z^T K_{ij}^T x^{k+1} - z^T (x^k \otimes I) =\\
 0 - z^T (x^k \otimes I) = -Z x^k = 0 \text{ } \forall x
\end{align*}  
But, this is a contradiction proving that $\oplus_{x \in \mathbb{R}^n} \mathcal{R}(N(x)) = \mathcal{N}(K_k)$. This together with $\text{vec}(P_k)^T N(x) =0$, $\forall x \in \mathbb{R}$ implies that
$\text{vec}(P_k) \in  \oplus_{x \in \mathbb{R}^n} \mathcal{R}(N(x))^\perp = \mathcal{N}(K_k)^\perp = \mathcal{R}(K_k^T)$ if and only if the symmetry condition holds. 
\end{proof}
\section{Proof of Lemmas \ref{invM_lemma} and \ref{invM_lemma2}} 
\label{InvM lem proof}
\begin{proof}
Suppose that $F_c^T+F_c$ is not negative definite. Since $F_c$ is Hurwitz, there exists a symmetric positive definite matrix $P_c \in \mathbb{R}^{n\times n}$ that satisfies the Laypunov equation $F_c^TP_c+P_cF_c+I=0$. Let  $W=P_c^{-1}>0$ and $T=\sqrt{P_c}$ be the transformation matrix. Multiplying the Lyapunov equation $F_c^TP_c+P_cF_c+I$ by $T^{-1}$ from the left and $T$ from the right yields $\hat{F}_c^T+\hat{F}_c+W=0$ where $\hat{F}_c=TF_cT^{-1}$. Otherwise, let $T=I$ (i.e., $\hat{F}_c=F_c$) and $W=-(F_c^T+F_c)>0$. Thus in either case, $\hat{F}_c^T+\hat{F}_c=-W$ and $\hat{M}_{k}= K_k (I_{m_k} \otimes \hat{F}_c^T) K_k^T$. Let $v$ be an arbitrary unit vector and $\sigma=\lambda_{\min} (W)/2>0$ where $\lambda_{\min} (W)$ denotes the smallest eigenvalue of $W$. Since $(\hat{M}_k+\sigma I)^T(\hat{M}_k+\sigma I) \geq 0$, it follows that
$$
v^T\hat{M}_k^T\hat{M}_kv \geq -\sigma (v^T\bar{M}_kv+\sigma)
$$
where $\bar{M}_k=-K_k (I_{m_k} \otimes W) K_k^T < 0$. Next we shall establish a lower bound on $|v^T\bar{M}_kv|$. Letting $z=K_k^Tv$, we have $\|z\|=1$ and
$$
|v^T\bar{M}_kv|=z^T ( I_{m_k} \otimes W ) z \geq \lambda_{\min} (I_{m_k} \otimes W)=\lambda_{\min} (W) 
$$
where the last equality follows from the fact that the eigenvalues of $I_{m_k} \otimes W$ are the same as those of $W$, each with multiplicity $m_k$. Thus
$$
\|\hat{M}_kv\|^2 = v^T\hat{M}_k^T\hat{M}_kv \geq \sigma (\lambda_{\min} (W)-\sigma) \geq \lambda_{\min}^2 (W)/4
$$
or $\|\hat{M}_kv\|\geq \alpha^{-1}:= \lambda_{\min} (W)/2 >0$ using that $\sigma=\lambda_{\min}(W)/2$. This clearly proves that $\hat{M}_k$ is invertible as a linear map from $\mathbb{R}^{\otimes k}=\{y \in \mathbb{R}^{m_k}: y=x^k, x\in \mathbb{R}^{n}\}$ to $\mathbb{R}^{\otimes k}$. A coordinate transformation (change of basis in $\mathbb{R}^{n}$) preserves its invertibility proving that $M_k$ is also invertible. Finally, $\|v\|=\|\hat{M}_k \hat{M}_k^{-1} v \| \geq \alpha^{-1} \|\hat{M}_k^{-1} v\|$, $\forall v \in \mathbb{R}^{m_{k+1}}$, implies $\|\hat{M}_k^{-1}\| \leq \alpha$.
\end{proof}
\end{appendices}

\bibliography{autosam}           

\begin{thebibliography}{32}
\providecommand{\natexlab}[1]{#1}
\providecommand{\url}[1]{\texttt{#1}}
\expandafter\ifx\csname urlstyle\endcsname\relax
  \providecommand{\doi}[1]{doi: #1}\else
  \providecommand{\doi}{doi: \begingroup \urlstyle{rm}\Url}\fi

\bibitem[Abu-Khalaf and Lewis(2005)]{abu2005nearly}
M.~Abu-Khalaf and F.~L. Lewis.
\newblock Nearly optimal control laws for nonlinear systems with saturating
  actuators using a neural network hjb approach.
\newblock \emph{Automatica}, 41\penalty0 (5):\penalty0 779--791, 2005.

\bibitem[Adurthi et~al.(2017)Adurthi, Singla, and Majji]{adurthi2017sparse}
N.~Adurthi, P.~Singla, and M.~Majji.
\newblock Sparse approximation--based collocation scheme for nonlinear optimal
  feedback control design.
\newblock \emph{Journal of Guidance, Control, and Dynamics}, 40\penalty0
  (2):\penalty0 248--264, 2017.

\bibitem[Al'Brekht(1961)]{al1961optimal}
E.~Al'Brekht.
\newblock On the optimal stabilization of nonlinear systems.
\newblock \emph{Journal of Applied Mathematics and Mechanics}, 25\penalty0
  (5):\penalty0 1254--1266, 1961.

\bibitem[Almubarak et~al.(2019)Almubarak, Sadegh, and
  Taylor]{AlmubarakSadeghandTaylor2019}
H.~Almubarak, N.~Sadegh, and D.~G. Taylor.
\newblock Infinite horizon nonlinear quadratic cost regulator.
\newblock In \emph{American Control Conference, 2019. Proceedings of the 2019}.
  IEEE, 2019.

\bibitem[Beard et~al.(1998)Beard, Saridis, and Wen]{beard1998approximate}
R.~W. Beard, G.~N. Saridis, and J.~T. Wen.
\newblock Approximate solutions to the time-invariant hamilton--jacobi--bellman
  equation.
\newblock \emph{Journal of Optimization theory and Applications}, 96\penalty0
  (3):\penalty0 589--626, 1998.

\bibitem[Beeler et~al.(2000)Beeler, Tran, and Banks]{beeler2000feedback}
S.~Beeler, H.~T. Tran, and H.~Banks.
\newblock Feedback control methodologies for nonlinear systems.
\newblock \emph{Journal of optimization theory and applications}, 107\penalty0
  (1):\penalty0 1--33, 2000.

\bibitem[Cimen(2008)]{cimen2008state}
T.~Cimen.
\newblock State-dependent riccati equation (sdre) control: A survey.
\newblock \emph{IFAC Proceedings Volumes}, 41\penalty0 (2):\penalty0
  3761--3775, 2008.

\bibitem[Flajolet and Sedgewick(2009)]{flajolet2009analytic}
P.~Flajolet and R.~Sedgewick.
\newblock \emph{Analytic combinatorics}.
\newblock cambridge University press, 2009.

\bibitem[Fujimoto and Sakamoto(2011)]{fujimoto2011stable}
R.~Fujimoto and N.~Sakamoto.
\newblock The stable manifold approach for optimal swing up and stabilization
  of an inverted pendulum with input saturation.
\newblock In \emph{IFAC world congress}, 2011.

\bibitem[Garrard(1972)]{garrard1972suboptimal}
W.~L. Garrard.
\newblock Suboptimal feedback control for nonlinear systems.
\newblock \emph{Automatica}, 8\penalty0 (2):\penalty0 219--221, 1972.

\bibitem[Garrard and Jordan(1977)]{garrard1977design}
W.~L. Garrard and J.~M. Jordan.
\newblock Design of nonlinear automatic flight control systems.
\newblock \emph{Automatica}, 13\penalty0 (5):\penalty0 497--505, 1977.

\bibitem[Garrard et~al.(1992)Garrard, Enns, and
  Antony~Snell]{garrard1992nonlinear}
W.~L. Garrard, D.~F. Enns, and S.~Antony~Snell.
\newblock Nonlinear feedback control of highly manoeuvrable aircraft.
\newblock \emph{International journal of control}, 56\penalty0 (4):\penalty0
  799--812, 1992.

\bibitem[Horibe and Sakamoto(2016)]{horibe2016swing}
T.~Horibe and N.~Sakamoto.
\newblock Swing up and stabilization of the acrobot via nonlinear optimal
  control based on stable manifold method.
\newblock \emph{IFAC-PapersOnLine}, 49\penalty0 (18):\penalty0 374--379, 2016.

\bibitem[Horn and Johnson(1994)]{horn1994topics}
R.~A. Horn and C.~R. Johnson.
\newblock \emph{Topics in matrix analysis}.
\newblock Cambridge university press, 1994.

\bibitem[Kalise and Kunisch(2018)]{kalise2018polynomial}
D.~Kalise and K.~Kunisch.
\newblock Polynomial approximation of high-dimensional
  hamilton--jacobi--bellman equations and applications to feedback control of
  semilinear parabolic pdes.
\newblock \emph{SIAM Journal on Scientific Computing}, 40\penalty0
  (2):\penalty0 A629--A652, 2018.

\bibitem[Khalil(2002)]{khalil2002nonlinear}
H.~K. Khalil.
\newblock \emph{Nonlinear systems}.
\newblock Prentice Hall, 2002.

\bibitem[Lawton and Beard(1998)]{lawton1998numerically}
J.~Lawton and R.~W. Beard.
\newblock Numerically efficient approximations to the hamilton-jacobi-bellman
  equation.
\newblock In \emph{American Control Conference, 1998. Proceedings of the 1998},
  volume~1, pages 195--199. IEEE, 1998.

\bibitem[Lewis et~al.(2012)Lewis, Vrabie, and Syrmos]{lewis2012optimal}
F.~L. Lewis, D.~Vrabie, and V.~L. Syrmos.
\newblock \emph{Optimal control}.
\newblock John Wiley \& Sons, 2012.

\bibitem[Liberzon(2011)]{liberzon2011calculus}
D.~Liberzon.
\newblock \emph{Calculus of variations and optimal control theory: a concise
  introduction}.
\newblock Princeton University Press, 2011.

\bibitem[Loparo and Blankenship(1978)]{loparo1978estimating}
K.~Loparo and G.~Blankenship.
\newblock Estimating the domain of attraction of nonlinear feedback systems.
\newblock \emph{IEEE Transactions on Automatic Control}, 23\penalty0
  (4):\penalty0 602--608, 1978.

\bibitem[Lukes(1969)]{lukes1969optimal}
D.~L. Lukes.
\newblock Optimal regulation of nonlinear dynamical systems.
\newblock \emph{SIAM Journal on Control}, 7\penalty0 (1):\penalty0 75--100,
  1969.

\bibitem[Lyshevski(1998)]{lyshevski1998optimal}
S.~E. Lyshevski.
\newblock Optimal control of nonlinear continuous-time systems: design of
  bounded controllers via generalized nonquadratic functionals.
\newblock In \emph{Proceedings of the 1998 American Control Conference. ACC
  (IEEE Cat. No. 98CH36207)}, volume~1, pages 205--209. IEEE, 1998.

\bibitem[Mayne(2014)]{mayne2014model}
D.~Q. Mayne.
\newblock Model predictive control: Recent developments and future promise.
\newblock \emph{Automatica}, 50\penalty0 (12):\penalty0 2967--2986, 2014.

\bibitem[Nishikawa et~al.(1971)Nishikawa, Sannomiya, and
  Itakura]{sannomiya1971method}
Y.~Nishikawa, N.~Sannomiya, and H.~Itakura.
\newblock A method for suboptimal design of nonlinear feedback systems.
\newblock \emph{Automatica}, 7\penalty0 (6):\penalty0 703--712, 1971.

\bibitem[Oishi and Sakamoto(2017)]{oishi2017numericalstablemani}
Y.~Oishi and N.~Sakamoto.
\newblock Numerical computational improvement of the stable-manifold method for
  nonlinear optimal control.
\newblock \emph{IFAC-PapersOnLine}, 50\penalty0 (1):\penalty0 5103--5108, 2017.

\bibitem[Qin and Badgwell(2003)]{qin2003survey}
S.~J. Qin and T.~A. Badgwell.
\newblock A survey of industrial model predictive control technology.
\newblock \emph{Control engineering practice}, 11\penalty0 (7):\penalty0
  733--764, 2003.

\bibitem[Sakamoto and van~der Schaft(2008)]{sakamoto2008analytical}
N.~Sakamoto and A.~J. van~der Schaft.
\newblock Analytical approximation methods for the stabilizing solution of the
  hamilton--jacobi equation.
\newblock \emph{IEEE Transactions on Automatic Control}, 53\penalty0
  (10):\penalty0 2335--2350, 2008.

\bibitem[Scheidemann(2005)]{scheidemann2005introduction}
V.~Scheidemann.
\newblock \emph{Introduction to complex analysis in several variables}.
\newblock Springer, 2005.

\bibitem[Tran et~al.(2017)Tran, Suzuki, and Sakamoto]{tran2017nonlinear}
A.~T. Tran, S.~Suzuki, and N.~Sakamoto.
\newblock Nonlinear optimal control design considering a class of system
  constraints with validation on a magnetic levitation system.
\newblock \emph{IEEE Control Systems Letters}, 1\penalty0 (2):\penalty0
  418--423, 2017.

\bibitem[Wernli and Cook(1975)]{wernli1975suboptimal}
A.~Wernli and G.~Cook.
\newblock Suboptimal control for the nonlinear quadratic regulator problem.
\newblock \emph{Automatica}, 11\penalty0 (1):\penalty0 75--84, 1975.

\bibitem[Xin and Balakrishnan(2005)]{xin2005new}
M.~Xin and S.~Balakrishnan.
\newblock A new method for suboptimal control of a class of non-linear systems.
\newblock \emph{Optimal Control Applications and Methods}, 26\penalty0
  (2):\penalty0 55--83, 2005.

\bibitem[Yoshida and Loparo(1989)]{yoshida1989quadratic}
T.~Yoshida and K.~A. Loparo.
\newblock Quadratic regulatory theory for analytic non-linear systems with
  additive controls.
\newblock \emph{Automatica}, 25\penalty0 (4):\penalty0 531--544, 1989.

\end{thebibliography}
\bibliographystyle{abbrvnat}

\end{document}